\documentclass[onefignum,onetabnum]{siamonline220329}
\usepackage{xcolor}

\usepackage{lipsum}
\usepackage{amsfonts}
\usepackage{graphicx}
\usepackage{epstopdf}
\usepackage{algorithmic}
\usepackage{mathtools}
\usepackage{amssymb}
\usepackage{mathrsfs}
\usepackage{multirow}
\usepackage{booktabs}
\usepackage{makecell}
\usepackage{tikz}
\usepackage{subfigure}
\usepackage{pgfplots}
\usepackage{scalefnt}
\usepgfplotslibrary{fillbetween}
\ifpdf
  \DeclareGraphicsExtensions{.eps,.pdf,.png,.jpg}
\else
  \DeclareGraphicsExtensions{.eps}
\fi

\definecolor{darkblue}{rgb}{0.0, 0.0, 0.55}
\definecolor{bordeaux}{rgb}{0.34, 0.01, 0.1}
\definecolor{color1}{RGB}{145,30,180}
\definecolor{color2}{RGB}{245,130,48}
\definecolor{color3}{RGB}{230,25,75}
\definecolor{lightgreen}{RGB}{144, 238, 144}

\usepackage{enumitem}
\setlist[enumerate]{leftmargin=.5in}
\setlist[itemize]{leftmargin=.5in}

\newsiamremark{remark}{Remark}
\newsiamremark{hypothesis}{Hypothesis}
\crefname{hypothesis}{Hypothesis}{Hypotheses}
\newsiamthm{claim}{Claim}

\def\bR{{\mathbb{R}}}

\def\C{{\mathbb{C}}}
\def\N{{\mathbb{N}}}

\def\x{{\mathbf{x}}}

\def\bu{{\mathbf{u}}}
\def\bv{{\mathbf{v}}}

\def\bc{{\mathbf{c}}}
\def\bd{{\mathbf{d}}}

\def\bH{{\mathbf{H}}}

\def\cB{{\mathcal{B}}}

\def\MM{{\mathbf{M}}}
\def\sC{{\mathscr{C}}}
\def\sM{{\mathscr{M}}}

\def\bi{\hbox{\bf{i}}}

\newcommand{\bra}[1]{\langle #1|}
\newcommand{\ket}[1]{|#1\rangle}

\def\NF{\hbox{\rm{NF}}}

\def\int{\hbox{\rm{int}}}

\def\sI{{\mathscr{I}}}
\def\sP{{\mathscr{P}}}

\headers{Scalable Ground-State Certification}{J. Wang et al.}

\title{Scalable Ground-State Certification of Quantum Spin Systems via Structured Noncommutative Polynomial Optimization\thanks{Submitted to the editors DATE.}}

\author{Jie Wang\thanks{State Key Laboratory of Mathematical Sciences, Academy of Mathematics and Systems Science, Chinese Academy of Sciences, Beijing, China (\email{wangjie212@amss.ac.cn}, \url{https://wangjie212.github.io/jiewang/}).}
\and David Jansen\thanks{ICFO-Institut de Ciencies Fotoniques, The Barcelona Institute of Science and Technology, Av. Carl Friedrich Gauss 3, 08860 Castelldefels (Barcelona), Spain (\email{djansenphysics@gmail.com})}
\and Ir\'en\'ee Fr\'erot\thanks{Laboratoire Kastler Brossel, Sorbonne Universit\'e, CNRS, ENS-PSL Research University, Coll\`ege de France, 4 Place Jussieu, 75005 Paris, France (\email{irenee.frerot@lkb.upmc.fr})}
\and Marc-Olivier Renou\thanks{Inria Paris-Saclay, Bâtiment Alan Turing, 1, rue Honoré d’Estienne d’Orves – 91120 Palaiseau, France (\email{marc-olivier.renou@inria.fr})}
\and Victor Magron\thanks{LAAS-CNRS \& Institute of Mathematics from Toulouse, France (\email{vmagron@laas.fr})}
\and Antonio Ac\'{\i}n\thanks{ICFO-Institut de Ciencies Fotoniques, The Barcelona Institute of Science and Technology, Av. Carl Friedrich Gauss 3, 08860 Castelldefels (Barcelona), Spain, and
ICREA - Instituci\'o Catalana de Recerca i Estudis Avan\c{c}ats, 08010 Barcelona, Spain (\email{antonio.acin@icfo.eu})}}

\usepackage{amsopn}
\DeclareMathOperator{\diag}{diag}

\begin{document}

\maketitle

\begin{abstract}
A fundamental challenge in quantum physics is determining the ground-state properties of many-body systems. Whereas standard variational approaches posit a wave-function ansatz and minimize over the possible states expressible by that ansatz, the problem can alternatively be formulated as a noncommutative polynomial optimization problem and treated through a hierarchy of semidefinite programming relaxations. In contrast to variational calculations, these relaxations provide lower bounds on ground-state energies and both lower and upper bounds on observable expectation values. However, this approach typically suffers from severe scalability issues, limiting its applicability to small-to-medium-scale systems. In this article, we demonstrate that systematically leveraging the inherent structures of the system can substantially mitigate these scalability challenges and thus permits computing meaningful bounds for quantum spin systems on square lattices of size up to $16\times16$.
\end{abstract}

\begin{keywords}
noncommutative polynomial optimization, structured semidefinite relaxation, sum of Hermitian squares, quantum spin system, ground-state energy
\end{keywords}

\begin{MSCcodes}
90C23, 81-08, 47N10
\end{MSCcodes}

\section{Introduction}
Determining the ground-state properties of quantum many-body systems is one of the most fundamental problems in condensed matter physics. Ground states are essential in understanding the low-temperature behavior of the system and can contain a variety of rich structures such as topological order~\cite{wen2017}, charge density waves~\cite{Gruner_88}, or superconductivity~\cite{bardeen_57}. Furthermore, when tuning the interaction strength between different degrees of freedom, various phases can emerge~\cite{sachdev2011quantum}, including ordered states or quantum spin liquids~\cite{savary_17}.
Several numerical methods have been established to compute ground states of quantum many-body systems in the past few decades.
Exact diagonalization of a Hamiltonian~\cite{sandvik_10} yields its exact ground state but is limited to small systems. Variational methods, including the density matrix renormalization group (DMRG)~\cite{WhiteDMRG}, projected entangled-pair states (PEPS)~\cite{Verstraete_04}, variational Monte Carlo (VMC), and neural-network quantum states (NNQS)~\cite{chooetal2019,Hibat-Allah_20}, use a parameterized wave function and optimize its parameters to minimize the energy. Quantum variational methods, such as the variational quantum eigensolver, also play an important role in determining ground states on quantum hardware~\cite{Peruzzo2014}. Although variational methods are highly effective in many important cases, their accuracy is essentially limited by the expressiveness of the ansatz~\cite{wu2023variational}. Quantum Monte Carlo (QMC) methods~\cite{sandvik2026high}, which simulate imaginary-time evolution, often encounter the sign problem for frustrated or fermionic systems~\cite{henelius2000sign,troyer2005computational}. Strictly speaking, all of the above methods except exact diagonalization essentially provide only upper bounds or estimates on the ground-state energy. Furthermore, after a candidate ground state has been obtained, it can be difficult to determine how close its energy is to the true ground-state energy and whether its properties, such as magnetization, accurately reflect those of the true ground state.

The semidefinite programming (SDP) relaxation method of~\cite{navascues2008convergent,pironio2010convergent}, also known as the Navascu\'es--Pironio--Ac\'\i n (NPA) hierarchy, through the lens of noncommutative polynomial
optimization stands out as a complementary approach by providing certified lower bounds on the
ground-state energy. At some relaxation order $d$, it retains expectation values of operator words of length at most $2d$, arranges them in a positive semidefinite (PSD) moment matrix, and imposes the algebraic relations among the operators as linear constraints. Every physical quantum state produces a feasible moment matrix, whereas a finite-order relaxation may also admit truncated moment data that do not arise from a physical state; minimizing over this enlarged set therefore gives a rigorous lower bound. Increasing $d$ adds longer words and tighter consistency constraints.

The framework was extended in~\cite{wang2024certifying} to provide lower and upper bounds on the ground-state expectation of any observable expressible as a polynomial in a family of basic observables. In \cite{mortimer_25}, it was then
demonstrated how the framework can be applied to bound observables acting on steady states
of open-system dynamics, and in \cite{jansen_25} it was used to study ground-state phase diagrams. See
also a series of other relevant development \cite{cho2024coarse,gao2025bootstrapping,han2020quantum,conformal,nancarrow2023bootstrapping,scheer2024hamiltonian,zhang2025bootstrapping} under the name of quantum
bootstrap methods. However, the relaxation method typically suffers from severe scalability
issues due to the rapidly growing size of SDPs, but \cite{wang2024certifying} also showed that the situation could
be greatly improved if various structures of the system are taken into account when building
SDP relaxations.
See \cite{baumgratz2012lower} for a pioneering work on the use of translational invariance to reduce SDP sizes. 

This work continues the line of research opened by \cite{wang2024certifying}. We show that by more carefully
exploiting the inherent structures of the considered system, both the scalability and the accuracy of the relaxation method could be further enhanced, reaching tighter bounds for the same
system sizes and treating larger systems. We demonstrate these improvements on four Heisenberg models:
the Heisenberg chain, the $J_1$--$J_2$ Heisenberg chain, the square-lattice Heisenberg model, and the square-lattice $J_1$--$J_2$ Heisenberg model. Our contributions are summarized as follows:

$\bullet$ In addition to the structures that were already exploited in~\cite{wang2024certifying}, we more thoroughly exploit
the sign symmetry, the conjugate symmetry, the permutation symmetry, and the dihedral
symmetry to further reduce the SDP size. We explain in detail how a dramatic size reduction can be achieved by fully exploiting those algebraic structures. Particularly in
the 2D case, we are able to perform a second round of block-diagonalization by exploiting the
translation symmetry.

$\bullet$ In \cite{wang2024certifying}, the positivity constraint on reduced density matrices is employed to strengthen
the SDP relaxations. We show that this positivity constraint has a block-diagonal structure due to the U(1)-symmetry. Moreover, we strengthen the SDP relaxations further by
imposing the state optimality conditions of~\cite{araujo2023first,fawzi2024certified}.

$\bullet$ Our ground-state energy bounds for the Heisenberg chain are notably more accurate than those in~\cite{wang2024certifying}. For square-lattice Heisenberg models, we can now treat lattices as large as $16\times16$, whereas \cite{wang2024certifying} was limited to the $10\times10$ lattice. The accuracy of the 2D results is also significantly improved.

$\bullet$ Complementing the physical perspective of \cite{wang2024certifying}, this article provides a rigorous exposition of all technical intricacies from a mathematical perspective, to lower the barrier to replication. Our long-term goal is to establish structured noncommutative polynomial optimization as one of the foundational tools for quantum many-body computation.

\section{Notation and preliminaries}
In this section, we collect some notation, definitions, and basic results that will be used in the rest of this article.

\subsection{Noncommutative polynomials}
Let $\N\coloneqq\{0,1,2,\cdots\}$. For $n\in\N\setminus\{0\}$, let $[n]\coloneqq\{1,2,\ldots,n\}$. We denote by $\bH_+^n$ the set of $n\times n$ PSD Hermitian matrices and write $A\succeq0$ to indicate that $A\in\bH_+^n$. For a matrix $A$ (resp. vector $\bv$), $A^*$ (resp. $\bv^*$) denotes its conjugate transpose. For a complex number $c$, $\bar{c}$ means its conjugate.
We use $I_n$ to stand for the $n\times n$ identity matrix. 
Let $\x=(x_1, \ldots, x_n)$ be a tuple of noncommuting variables. The set of all finite words, or monomials, in $\x$ is denoted by $\langle\x\rangle$. The empty word is
denoted by $1$. The degree of a monomial $w\in\langle\x\rangle$, denoted by $\deg(w)$, is the length of $w$ as a word. We denote by $\C\langle\x\rangle$ the ring of complex polynomials in the noncommuting variables $\x$. An element $f\in\C\langle\x\rangle$ can be written as $f=\sum_{w\in\langle\x\rangle}f_ww$ with $f_w\in\C$, which is called a {\em noncommutative polynomial}, and the degree of $f$ is defined as $\deg(f)\coloneqq\max\{\deg(w):f_w\ne0\}$. The ring $\C\langle\x\rangle$ is equipped with the involution $\star$ that fixes $\{x_1, \ldots, x_n\}$ pointwise, maps coefficients to their complex conjugates, and reverses words, so that $\C\langle\x\rangle$ is the $\star$-algebra freely generated by $n$ symmetric letters $x_1, \ldots, x_n$. For $d\in\N$, let $W_d$ be the column vector consisting of all words of length at most $d$ arranged with respect to the lexicographic order.
A noncommutative polynomial $f$ is called a {\em sum of Hermitian squares (SOHS)} if there exist $g_1, \ldots, g_t\in\C\langle\x\rangle$ such that $f = g_1^{\star}g_1+g_2^{\star}g_2+\cdots+g_t^{\star}g_t$. Deciding whether a given
noncommutative polynomial $f$ is a SOHS can be cast as an SDP due to the following lemma.
\begin{lemma}[\cite{helton2002positive}, Lemma 2.1]
Let $f\in\C\langle\x\rangle$ with $\deg(f)=2d$. Then $f$ is a SOHS if and only if there exists $G\in\bH_+^{|W_d|}$ ($|W_d|$ stands for the dimension of $W_d$) satisfying
\begin{equation}\label{sec2-eq1}
f=W_d^{\star}GW_d,
\end{equation}
where $W_d^{\star}$ is the row vector consisting of $u^{\star}, u\in W_d$.
\end{lemma}

\subsection{The Pauli algebra for quantum $\frac{1}{2}$-spin systems}
The quantum state of a spin-$\frac{1}{2}$ particle is captured by the Pauli matrices $\sigma^x,\sigma^y,\sigma^z$ which are defined as
\begin{equation}
    \sigma^x=\begin{bmatrix}0&1\\1&0\end{bmatrix},\quad
     \sigma^y=\begin{bmatrix}0&-\bi\\\bi&0\end{bmatrix},\quad
      \sigma^z=\begin{bmatrix}1&0\\0&-1\end{bmatrix}.
\end{equation}
Fix an $N\in\N\setminus\{0\}$. We consider a quantum spin system consisting of $N$ spin-$\frac{1}{2}$ particles that are located on some lattice. The $N$-site Pauli algebra $\sP_N$ associated with this quantum spin system is the algebra generated by Pauli matrices $\{\sigma_i^a\}_{i\in[N],a\in\{x,y,z\}}$ (where each $\sigma_i^a$ is the Pauli matrix $\sigma^a$ acting on the $i$-th site) over $\C$. The algebraic relations satisfied by Pauli matrices define a two-sided ideal in the noncommutative polynomial ring $\C\left\langle\{\sigma_i^a\}_{i\in[N],a\in\{x,y,z\}}\right\rangle$\footnote{We slightly abuse notation by regarding $\{\sigma_i^a\}_{i\in[N],a\in\{x,y,z\}}$ as a tuple of noncommuting variables here.}:
\begin{equation}
    \begin{split}
    \sI_N\coloneqq\Big\langle&(\sigma_i^a)^2-1,\sigma_i^x\sigma_i^y-\bi\sigma_i^z,\sigma_i^y\sigma_i^z-\bi\sigma_i^x,\sigma_i^z\sigma_i^x-\bi\sigma_i^y,\sigma_i^a\sigma_i^b+\sigma_i^b\sigma_i^a,\sigma_i^b\sigma_j^c-\sigma_j^c\sigma_i^b\\
    &:\,1\le i\ne j\le N,a\ne b\in\{x,y,z\},c\in\{x,y,z\}\Big\rangle.
    \end{split}
\end{equation}
These are the equalities resulting from the standard Pauli relations at each site $i$ and the commutation between operators at different sites, $i$ and $j$.
Because of these equalities, it holds that $\sP_N=\C\left\langle\{\sigma_i^a\}_{i\in[N],a\in\{x,y,z\}}\right\rangle/\sI_N$.

The algebraic relations also give rise to the following replacement rules on monomials in $\sP_N$:
\begin{subequations}\label{rules}
\begin{align}
    (\sigma_i^a)^2\quad&\longrightarrow\quad1,\quad i=1,\ldots,N,a\in\{x,y,z\},\\
    \sigma_i^x\sigma_i^y\quad&\longrightarrow\quad\bi\sigma_i^z,\quad i=1,\ldots,N,\label{rule1}\\
    \sigma_i^y\sigma_i^x\quad&\longrightarrow\quad-\bi\sigma_i^z,\quad i=1,\ldots,N,\\
    \sigma_i^y\sigma_i^z\quad&\longrightarrow\quad\bi\sigma_i^x,\quad i=1,\ldots,N,\\
    \sigma_i^z\sigma_i^y\quad&\longrightarrow\quad-\bi\sigma_i^x,\quad i=1,\ldots,N,\\
    \sigma_i^z\sigma_i^x\quad&\longrightarrow\quad\bi\sigma_i^y,\quad i=1,\ldots,N,\\
    \sigma_i^x\sigma_i^z\quad&\longrightarrow\quad-\bi\sigma_i^y,\quad i=1,\ldots,N,\label{rule6}\\
    \sigma_i^a\sigma_j^b\quad&\longrightarrow\quad\sigma_j^b\sigma_i^a,\quad 1\le i\ne j\le N,a,b\in\{x,y,z\}.
\end{align}
\end{subequations}
By iteratively applying the above replacement rules, one can reduce any monomial $w\in\sP_N$ to the unique {\em normal form} $\NF(w)\coloneqq c\sigma_{i_1}^{a_1}\sigma_{i_2}^{a_2}\cdots\sigma_{i_r}^{a_r}$ with $c\in\{1,-1,\bi,-\bi\}$, $1\le i_1<i_2<\cdots<i_r\le N$. Accordingly, we denote the coefficient $c$ by $\sC(w)$ and the monomial $\sigma_{i_1}^{a_1}\sigma_{i_2}^{a_2}\cdots\sigma_{i_r}^{a_r}$ by $\sM(w)$. Let 
\begin{equation}
    \tilde{W}_d\coloneqq\bigcup_{r=0}^d\left\{\sigma_{i_1}^{a_1}\sigma_{i_2}^{a_2}\cdots\sigma_{i_r}^{a_r}
\,\middle|\,1\le i_1<i_2<\cdots<i_r\le N,a_k\in\{x,y,z\},k\in[N]\right\}
\end{equation}
for any $d\in\N$. Then, any noncommutative polynomial $f\in\sP_N$ of degree $d$ can be written as $f=\sum_{w\in\tilde{W}_d}c_ww$ with $c_w\in\C$.

There is a $\star$-isomorphism $\pi$ from $\sP_N$ to the matrix algebra $\C^{2^N\times 2^N}$, defined by
\begin{equation}\label{pauli-representation}
\pi(\sigma_{i}^{a})=\underbrace{I_2\otimes\cdots\otimes I_2}_{i-1}\otimes\sigma^a\otimes \underbrace{I_2\otimes\cdots\otimes I_2}_{N-i}.
\end{equation}
It can be shown that $\{\pi(\sigma_{i_1}^{a_1}\sigma_{i_2}^{a_2}\cdots\sigma_{i_r}^{a_r})\}_{1\le i_1<i_2<\cdots<i_r\le N, a_k\in\{x,y,z\}, k\in[N]}$ form a linear basis of $\C^{2^N\times 2^N}$.
We say that an element $f\in\sP_N$ is \emph{nonnegative}, denoted by $f\succeq0$, if $\pi(f)\succeq0$.
Furthermore, we say that $f$ is a SOHS if $f\equiv\sum_{i=1}^tg_i^{\star}g_i$ for some $g_1,\ldots,g_t\in\sP_N$.

\begin{proposition}
Let $f\in\sP_N$. Then $f$ is a SOHS if and only if there exists some $d\in\N$ and $G\in\bH_+^{|\tilde{W}_d|}$ such that $f\equiv \tilde{W}_d^{\star}G\tilde{W}_d$, where $\tilde{W}_d^{\star}$ is the row vector consisting of $u^{\star}, u\in \tilde{W}_d$.
\end{proposition}

It turns out that for elements in $\sP_N$, being nonnegative is equivalent to being a SOHS.
\begin{theorem}\label{thm1}
Let $f\in\sP_N$. Then $f$ is nonnegative if and only if $f$ is a SOHS.
\end{theorem}
\begin{proof}
We prove only the ``only if'' direction, since the converse is immediate.
Assume that $f\succeq0$. Because $\pi(f)\in\bH_+^{2^N}$, there exists a matrix $A\in\C^{2^N\times 2^N}$ such that $\pi(f)=A^*A$. Hence, $f=\pi^{-1}(A)^{\star}\pi^{-1}(A)$.
\end{proof}

\section{Ground-state certification via noncommutative polynomial optimization}\label{sec:ground-certification}

Suppose that $H\in\sP_N$ represents the Hamiltonian of a quantum spin system via the $\star$-isomorphism $\pi$ defined in Eq.~\eqref{pauli-representation}. Equivalently, $\pi(H)$ is the usual $2^N\times2^N$ matrix of the Hamiltonian. The ground-state energy $E_{\mathrm{GS}}$ of the system is the smallest eigenvalue of $\pi(H)$ and can therefore be computed through the following noncommutative polynomial optimization problem:
\begin{equation}\label{ncpop}
    \begin{aligned}E_{\rm GS} \coloneqq\min\limits_{\{\ket{\psi},\sigma_i^a\}}&\quad \bra\psi H \ket\psi\\
    \text{subject to:}&\quad\langle\psi|\psi\rangle=1,\\
    &\quad(\sigma_i^a)^2=1,\quad i=1,\ldots,N, a\in\{x,y,z\},\\
    &\quad\sigma_i^x\sigma_i^y=\bi\sigma_i^z,\quad\sigma_i^y\sigma_i^x=-\bi\sigma_i^z,\quad i=1,\ldots,N,\\
    &\quad\sigma_i^y\sigma_i^z=\bi\sigma_i^x,\quad\sigma_i^z\sigma_i^y=-\bi\sigma_i^x,\quad i=1,\ldots,N,\\
    &\quad\sigma_i^z\sigma_i^x=\bi\sigma_i^y,\quad\sigma_i^x\sigma_i^z=-\bi\sigma_i^y,\quad i=1,\ldots,N,\\
    &\quad\sigma_i^a\sigma_j^b=\sigma_j^b\sigma_i^a,\quad 1\le i\ne j\le N,a,b\in\{x,y,z\}.
    \end{aligned}
\end{equation}
In Eq.~\eqref{ncpop}, both the state $\ket{\psi}$ and the operators $\{\sigma_i^a\}$ are decision variables: the latter range over operator representations satisfying the displayed Pauli relations. This formulation is equivalent to the standard physics formulation with fixed tensor-product Pauli matrices. Indeed, the $N$-site Pauli algebra is $\star$-isomorphic to $\C^{2^N\times2^N}$ and has, up to unitary equivalence, a unique irreducible representation; any other representation is a direct sum of copies of it and cannot change the minimum.
An equivalent reformulation of \eqref{ncpop} is
\begin{equation}\label{ncpop-d}
\begin{aligned}\max\limits_{\lambda}&\quad \lambda\\
    \text{subject to:}&\quad H-\lambda\succeq0.
    \end{aligned}
\end{equation}
By replacing the nonnegativity constraint with a SOHS decomposition, we obtain an SDP strengthening of \eqref{ncpop-d} for each $d\in\N\setminus\{0\}$:
\begin{equation}\label{ncsos}
\begin{aligned}\lambda_d\coloneqq\max\limits_{\lambda,G}&\quad \lambda\\
    \text{subject to:}&\quad H-\lambda\equiv \tilde{W}_d^{\star}G\tilde{W}_d,\\
    &\quad G\succeq0.
    \end{aligned}
\end{equation}
The optimum $\lambda_d$ of \eqref{ncsos} provides a lower bound on the ground-state energy $E_{\mathrm{GS}}$ and becomes tighter as $d$ grows:
\begin{equation}
\label{lbgs}
\cdots\le\lambda_d\le\lambda_{d+1}\le\cdots\le E_{\mathrm{GS}}.
\end{equation}
Moreover, Theorem \ref{thm1} implies that $\lambda_d=E_{\mathrm{GS}}$ for some $d\in\N$.

The reformulation \eqref{ncpop-d} as well as the SOHS strengthening
\eqref{ncsos} has a dual side. Let $\ell:\sP_N\rightarrow\C$ be a linear functional acting on the quotient ring $\sP_N$. Let $\Sigma_N$ denote the set of SOHS polynomials in $\sP_N$ and $\tilde{W}=\bigcup_{d=0}^{\infty}\tilde{W}_d$. Then the dual of \eqref{ncpop-d} reads as
\begin{equation}\label{state}
\begin{aligned}
    \min\limits_{\ell} &\quad \ell(H)\\
    \text{subject to:} &\quad \ell(f) \geq 0, \quad \forall f \in \Sigma_N,\\
    &\quad\ell(u^{\star})=\overline{\ell(u)},\quad \forall u\in\tilde{W},\\
    & \quad \ell(1)=1.
\end{aligned}
\end{equation}
For a linear functional $\ell:[\sP_N]_{2d} (\coloneqq\{f\in\sP_N\mid\deg(f)\le2d\})\rightarrow\C$, define the \emph{moment matrix} $\mathbf{M}_{d}(\ell)$ indexed by $\tilde{W}_d$ through $[\MM_{d}]_{v,w}=\ell(v^{\star}w)$. Then the $d$-th order moment relaxation for \eqref{ncpop} is given by
\begin{equation}\label{mom}
\begin{aligned}
\lambda_d^*\coloneqq\min\limits_{\{\ell(u)\}_{u\in\tilde{W}_{2d}}}&\quad\ell(H)\\
    \text{subject to:}&\quad\MM_{d}(\ell)\succeq0, \\
    &\quad\MM_{d}(\ell) \text{ obeys the monomial replacement rules \eqref{rules}},\\
    &\quad\ell(u^{\star})=\overline{\ell(u)},\quad \forall u\in\tilde{W}_{2d},\\
    &\quad\ell(1)=1.
\end{aligned}
\end{equation}
The program \eqref{mom} is the dual SDP of \eqref{ncsos}.
The primal--dual SDP pair \eqref{ncsos}--\eqref{mom} composes the well-known NPA hierarchy~\cite{navascues2008convergent,pironio2010convergent}. To see its meaning, take any polynomial $p=\sum_{u\in\tilde W_d}p_u u$. The constraint $\MM_d(\ell)\succeq0$ is equivalent to $\ell(p^\star p)\ge0$ for every such $p$, while the linear constraints identify entries that are equal under the Pauli relations. Every density matrix $\rho$ defines a feasible functional $\ell_\rho(f)=\operatorname{tr}(\rho\pi(f))$, but the converse need not hold at finite $d$; this enlargement of the feasible set explains why \eqref{mom} yields a lower bound. On the dual side, \eqref{ncsos} certifies the same bound by writing $H-\lambda$ as a SOHS modulo the Pauli relations. Because all diagonal entries of $\MM_d(\ell)$ equal one, there is no duality gap between \eqref{ncsos} and \eqref{mom}, so $\lambda_d=\lambda_d^*$~\cite{josz2016strong}. As $d$ increases, both sides use a larger word space and yield the monotone sequence in Eq.~\eqref{lbgs}. 

Apart from providing a lower bound $E_{\mathrm{lb}}$ on the ground-state energy $E_{\mathrm{GS}}$, the NPA hierarchy also allows one to go beyond
energies and derive lower and upper bounds on the ground-state expectation value $O_{\mathrm{GS}}$ of any observable $O$~\cite{wang2024certifying}. For that, one makes use also of an upper bound $E_{\mathrm{ub}}$ on the ground-state energy, obtained for example by a variational method, and solves the following two SDPs:
\begin{equation}\label{mom1}
\begin{aligned}
\min/\max&\quad\ell(O)\\
    \text{subject to:}&\quad\MM_{d}(\ell)\succeq0, \\
    &\quad\MM_{d}(\ell) \text{ obeys the monomial replacement rules \eqref{rules}},\\
    &\quad\ell(u^{\star})=\overline{\ell(u)},\quad \forall u\in\tilde{W}_{2d},\\
    &\quad\ell(1)=1,\\
 &\quad E_{\mathrm{lb}}\le\ell(H)\le E_{\mathrm{ub}}.\\
\end{aligned}
\end{equation}
The minima and maxima of these programs provide the desired lower and upper bounds on $O_{\mathrm{GS}}$, respectively.

\begin{remark}
If symmetries are exploited to simplify the SDP \eqref{mom1} (as will be done in the next section), the optimization becomes effectively restricted to the subspace of symmetric ground states. Nevertheless, for finite-size systems, this generally does not impose a genuine restriction, as the ground state is typically unique except at certain critical points.
\end{remark}

\section{Reducing SDP sizes by exploiting structures}\label{structure}
From now on, we consider the 1D and 2D Heisenberg models under periodic boundary conditions. We shall show how to dramatically reduce the size of \eqref{ncsos}-\eqref{mom} by exploiting rich algebraic structures of the Heisenberg models. The Hamiltonians of the four Heisenberg models under consideration are given below.

$\bullet$ The Heisenberg chain:
\begin{equation}\label{model1}
    H = \frac{1}{4}\sum_{i=1}^N \sum_{a\in\{x,y,z\}} \sigma_i^a \sigma_{i+1}^a;
\end{equation}

$\bullet$ The $J_1$--$J_2$ Heisenberg chain ($J_1=1$):
\begin{equation}\label{model2}
    H = \frac{1}{4}\sum_{i=1}^N \sum_{a\in\{x,y,z\}} \left(\sigma_i^a \sigma_{i+1}^a + J_2 \sigma_i^a \sigma_{i+2}^a\right);
\end{equation}

$\bullet$ The square-lattice Heisenberg model:
\begin{equation}\label{model3}
    H = \frac{1}{4}\sum_{i=1}^L \sum_{j=1}^L \sum_{a\in\{x,y,z\}} \sigma_{(i,j)}^a \left(\sigma_{(i+1,j)}^a + \sigma_{(i,j+1)}^a\right);
\end{equation}

$\bullet$ The square-lattice $J_1$--$J_2$ Heisenberg model ($J_1=1$):
\begin{equation}\label{model4}
    H = \frac{1}{4}\sum_{i=1}^L \sum_{j=1}^L \sum_{a\in\{x,y,z\}} \sigma_{(i,j)}^a \left(\sigma_{(i+1,j)}^a + \sigma_{(i,j+1)}^a + J_2\left(\sigma_{(i+1,j+1)}^a + \sigma_{(i+1,j-1)}^a\right)\right).
\end{equation}

In the 1D case, the number of spins is $N=L$; in the 2D case, the number of spins is $N=L^2$.

\subsection{Sparsity}
Because the Hamiltonians of the Heisenberg models contain only monomials supported on contiguous sites, we may build the moment-SOHS relaxation \eqref{ncsos}--\eqref{mom} by
using a sparse monomial basis which consists of monomials supported on contiguous sites. Specifically, for 1D Heisenberg models, we use the following sparse monomial basis:
\begin{equation}
    \cB_d = \bigcup_{j=0}^d\left\{\sigma_i^{a_1}\sigma_{i+1}^{a_2}\cdots\sigma_{i+j-1}^{a_j}\,\middle|\, i\in[L],a_k\in\{x,y,z\},k\in[j]\right\}
\end{equation}
instead of $\tilde{W}_d$. 
Moreover, to capture long-range correlations, we also include the monomials $\{\sigma_{i}^a\sigma_{i+j}^b\}_{a,b\in\{x,y,z\},j=2,\ldots,r}$ in $\cB_d$ for a given $r\ge2$ and $d\ge2$.
For 2D Heisenberg models, we use the following sparse monomial basis:
\begin{align*}
    \cB_1&=\{1\}\cup\left\{\sigma_{(i,j)}^a\,\middle|\,i,j \in [L],a\in\{x,y,z\}\right\},\\
    \cB_2&=\cB_1\cup\left\{\sigma_{(i,j)}^a\sigma_{(i+s,j+t)}^b\,\middle|\, i,j \in [L],-4\le s,t\le4,a,b\in\{x,y,z\}\right\},\\
    \cB_3&=\cB_2\cup\left\{\sigma_{(i,j)}^a \sigma_{(i+s_1,j+t_1)}^b \sigma_{(i+s_2,j+t_2)}^c\,\middle|\, i,j \in [L],(s_1,t_1,s_2,t_2)\in\mathcal{T},a,b,c\in\{x,y,z\}\right\},\\
    \cB_4&=\cB_3\cup\left\{\sigma_{(i,j)}^a \sigma_{(i+1,j)}^b \sigma_{(i,j+1)}^c \sigma_{(i+1,j+1)}^d\,\middle|\, i,j \in [L],a,b,c,d\in\{x,y,z\}\right\},
\end{align*}
where $\mathcal{T}=\{(1,0,2,0),(0,1,0,2),(0,-1,1,-1),(0,1,1,1),(1,0,1,-1),(1,0,1,1)\}$ (see Fig.~\ref{threebody}).

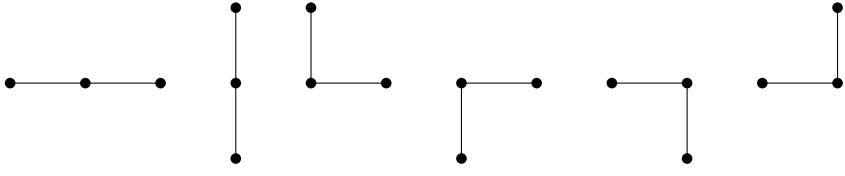
\begin{figure}\label{threebody}
\centering
\begin{tikzpicture}
\draw (0,0)--(1,0)--(2,0);
\fill (0,0) circle (2pt);
\fill (1,0) circle (2pt);
\fill (2,0) circle (2pt);
\draw (3,-1)--(3,0)--(3,1);
\fill (3,-1) circle (2pt);
\fill (3,0) circle (2pt);
\fill (3,1) circle (2pt);
\draw (10,0)--(11,0)--(11,1);
\fill (10,0) circle (2pt);
\fill (11,0) circle (2pt);
\fill (11,1) circle (2pt);
\draw (8,0)--(9,0)--(9,-1);
\fill (8,0) circle (2pt);
\fill (9,0) circle (2pt);
\fill (9,-1) circle (2pt);
\draw (6,-1)--(6,0)--(7,0);
\fill (6,-1) circle (2pt);
\fill (6,0) circle (2pt);
\fill (7,0) circle (2pt);
\draw (4,1)--(4,0)--(5,0);
\fill (4,1) circle (2pt);
\fill (4,0) circle (2pt);
\fill (5,0) circle (2pt);
\end{tikzpicture}
\caption{Lattice sites of three-body monomials appearing in the monomial basis for the 2D case.}
\end{figure}

\subsection{Symmetry}
We continue to reduce the size of \eqref{ncsos}-\eqref{mom} by exploiting various inherent symmetries of the Heisenberg models.
\subsubsection{Sign symmetry of the model}\label{ss1}
The Heisenberg models possess certain sign symmetries. Actually, they are invariant under the following sign flips:
\begin{subequations}\label{eq:ss1}
\begin{align}
    s_{xy} : & \, (\sigma_i^x,\sigma_i^y,\sigma_i^z)_{i=1}^N\longrightarrow (-\sigma_i^x,-\sigma_i^y,\sigma_i^z)_{i=1}^N,\\
    s_{yz} : & \, (\sigma_i^x,\sigma_i^y,\sigma_i^z)_{i=1}^N\longrightarrow (\sigma_i^x,-\sigma_i^y,-\sigma_i^z)_{i=1}^N,\\
    s_{zx} : & \, (\sigma_i^x,\sigma_i^y,\sigma_i^z)_{i=1}^N\longrightarrow (-\sigma_i^x,\sigma_i^y,-\sigma_i^z)_{i=1}^N.
\end{align}
\end{subequations}
These sign symmetries induce a block-diagonal structure in the moment matrix $\MM_d(\ell)$, as we now explain.
For a monomial $u\in\sP_N$, we define its \emph{signature} with respect to the sign flips~\eqref{eq:ss1} as 
\begin{equation}
    \xi(u)\coloneqq(s_{xy}(u)/u, s_{yz}(u)/u, s_{zx}(u)/u) \in \{-1, 1\}^3.
\end{equation}
Since $s_{zx}$ is the composition of $s_{xy}$ and $s_{yz}$, there are four possible signatures $\xi(u)$ for $u\in\sP_N$: $(1,1,1),(1,-1,-1),(-1,1,-1),(-1,-1,1)$. The monomial basis $\cB_d$ can therefore be partitioned into four subbases $\cB_d^{(1)}, \cB_d^{(2)}, \cB_d^{(3)}, \cB_d^{(4)}$, corresponding to these four signatures, respectively.

\begin{proposition}\label{prop1}
In the moment relaxation \eqref{mom} for a Heisenberg model, there is no loss of generality in assuming that $\ell(u)=0$ whenever $\xi(u)\ne(1,1,1)$. Consequently, after appropriate row and column permutations, the moment matrix $\MM_d(\ell)$ is block diagonal, with four nonzero blocks indexed by $\cB_d^{(1)}, \cB_d^{(2)}, \cB_d^{(3)}, \cB_d^{(4)}$, respectively.
\end{proposition}
\begin{proof}
For any feasible solution $\ell$ to \eqref{mom}, we define a linear functional $\tilde{\ell}:[\sP_N]_{2d}\rightarrow\C$ by
\begin{equation}
    \tilde{\ell}(u)=\frac{1}{4}\left(\ell(u)+\ell(s_{xy}(u))+\ell(s_{yz}(u))+\ell(s_{zx}(u))\right), \quad \forall u\in\tilde{W}_{2d}.
\end{equation}
It is readily verified that (i) $\tilde{\ell}(H)=\ell(H)$, (ii) $\MM_d(\tilde{\ell})\succeq0$, and (iii) $\tilde{\ell}(u)=0$ whenever $\xi(u)\ne(1,1,1)$. Thus, $\tilde{\ell}$ is feasible for \eqref{mom} and has the same objective value as $\ell$, proving the first claim. For $v,w\in\cB_d$, we have $\xi(v^{\star}w)=(1,1,1)$ if and only if $\xi(v)=\xi(w)$. Appropriate row and column permutations therefore expose the four claimed nonzero blocks of $\MM_d(\ell)$.
\end{proof}

For the 1D case and $d=4$, the partition of $\cB_d$ is given in \cref{tab:sign2}. The 2D case is similar.

\begin{table}[!ht]
\caption{Subbases indexing the sign symmetry blocks for $d = 4$.}
\label{tab:sign2}
\centering
\begin{tabular}{c|c}
\Xhline{1pt}
Subbasis& Monomials\\
\Xhline{1pt}
\multirow{3}{*}{$\cB_d^{(1)}$}& $1,\sigma_i^a\sigma_{i+j}^a,\sigma_i^a\sigma_{i+1}^b\sigma_{i+2}^c,\sigma_i^a\sigma_{i+1}^a\sigma_{i+2}^a\sigma_{i+3}^a$,\\
&$\sigma_i^a\sigma_{i+1}^a\sigma_{i+2}^b\sigma_{i+3}^b,\sigma_i^a\sigma_{i+1}^b\sigma_{i+2}^a\sigma_{i+3}^b$,$\sigma_i^a\sigma_{i+1}^b\sigma_{i+2}^b\sigma_{i+3}^a$,\\
&$a\ne b\ne c\in\{x,y,z\},i\in[L],j\in[r]$\\
\hline
\multirow{6}{*}{$\cB_d^{(2)}$}& $\sigma_i^z,\sigma_i^a\sigma_{i+j}^b,\sigma_i^z\sigma_{i+1}^z\sigma_{i+2}^z,\sigma_i^z\sigma_{i+1}^a\sigma_{i+2}^a$,\\
&$\sigma_i^a\sigma_{i+1}^z\sigma_{i+2}^a,\sigma_i^a\sigma_{i+1}^a\sigma_{i+2}^z,\sigma_i^z\sigma_{i+1}^z\sigma_{i+2}^a\sigma_{i+3}^b$,\\
&$\sigma_i^z\sigma_{i+1}^a\sigma_{i+2}^z\sigma_{i+3}^b,\sigma_i^z\sigma_{i+1}^a\sigma_{i+2}^b\sigma_{i+3}^z$,$\sigma_i^a\sigma_{i+1}^z\sigma_{i+2}^z\sigma_{i+3}^b$,\\
&$\sigma_i^a\sigma_{i+1}^z\sigma_{i+2}^b\sigma_{i+3}^z,\sigma_i^a\sigma_{i+1}^b\sigma_{i+2}^z\sigma_{i+3}^z$,$\sigma_i^a\sigma_{i+1}^b\sigma_{i+2}^b\sigma_{i+3}^b$,\\
&$\sigma_i^b\sigma_{i+1}^a\sigma_{i+2}^b\sigma_{i+3}^b,\sigma_i^b\sigma_{i+1}^b\sigma_{i+2}^a\sigma_{i+3}^b$,$\sigma_i^b\sigma_{i+1}^b\sigma_{i+2}^b\sigma_{i+3}^a$,\\
&$a\ne b\in\{x,y\},i\in[L],j\in[r]$\\
\hline
\multirow{6}{*}{$\cB_d^{(3)}$}& $\sigma_i^x,\sigma_i^a\sigma_{i+j}^b,\sigma_i^x\sigma_{i+1}^x\sigma_{i+2}^x,\sigma_i^x\sigma_{i+1}^a\sigma_{i+2}^a$,\\
&$\sigma_i^a\sigma_{i+1}^x\sigma_{i+2}^a,\sigma_i^a\sigma_{i+1}^a\sigma_{i+2}^x,\sigma_i^x\sigma_{i+1}^x\sigma_{i+2}^a\sigma_{i+3}^b$,\\
&$\sigma_i^x\sigma_{i+1}^a\sigma_{i+2}^x\sigma_{i+3}^b,\sigma_i^x\sigma_{i+1}^a\sigma_{i+2}^b\sigma_{i+3}^x$,$\sigma_i^a\sigma_{i+1}^x\sigma_{i+2}^x\sigma_{i+3}^b$,\\
&$\sigma_i^a\sigma_{i+1}^x\sigma_{i+2}^b\sigma_{i+3}^x,\sigma_i^a\sigma_{i+1}^b\sigma_{i+2}^x\sigma_{i+3}^x$,$\sigma_i^a\sigma_{i+1}^b\sigma_{i+2}^b\sigma_{i+3}^b$,\\
&$\sigma_i^b\sigma_{i+1}^a\sigma_{i+2}^b\sigma_{i+3}^b,\sigma_i^b\sigma_{i+1}^b\sigma_{i+2}^a\sigma_{i+3}^b$,$\sigma_i^b\sigma_{i+1}^b\sigma_{i+2}^b\sigma_{i+3}^a$,\\
&$a\ne b\in\{y,z\},i\in[L],j\in[r]$\\
\hline
\multirow{6}{*}{$\cB_d^{(4)}$}& $\sigma_i^y,\sigma_i^a\sigma_{i+j}^b,\sigma_i^y\sigma_{i+1}^y\sigma_{i+2}^y,\sigma_i^y\sigma_{i+1}^a\sigma_{i+2}^a$,\\
&$\sigma_i^a\sigma_{i+1}^y\sigma_{i+2}^a,\sigma_i^a\sigma_{i+1}^a\sigma_{i+2}^y,\sigma_i^y\sigma_{i+1}^y\sigma_{i+2}^a\sigma_{i+3}^b$,\\
&$\sigma_i^y\sigma_{i+1}^a\sigma_{i+2}^y\sigma_{i+3}^b,\sigma_i^y\sigma_{i+1}^a\sigma_{i+2}^b\sigma_{i+3}^y$,$\sigma_i^a\sigma_{i+1}^y\sigma_{i+2}^y\sigma_{i+3}^b$,\\
&$\sigma_i^a\sigma_{i+1}^y\sigma_{i+2}^b\sigma_{i+3}^y,\sigma_i^a\sigma_{i+1}^b\sigma_{i+2}^y\sigma_{i+3}^y$,$\sigma_i^a\sigma_{i+1}^b\sigma_{i+2}^b\sigma_{i+3}^b$,\\
&$\sigma_i^b\sigma_{i+1}^a\sigma_{i+2}^b\sigma_{i+3}^b,\sigma_i^b\sigma_{i+1}^b\sigma_{i+2}^a\sigma_{i+3}^b$,$\sigma_i^b\sigma_{i+1}^b\sigma_{i+2}^b\sigma_{i+3}^a$,\\
&$a\ne b\in\{x,z\},i\in[L],j\in[r]$\\
\Xhline{1pt}
\end{tabular}
\end{table}

From now on, we assume that $\ell(u)=0$ whenever $\xi(u)\ne(1,1,1)$ in the moment relaxation~\eqref{mom}.

\subsubsection{Sign symmetry of the Hamiltonian}\label{ss2}
In addition to the sign symmetries of the whole model, the Hamiltonians of the Heisenberg models possess extra sign symmetries:
\begin{subequations}\label{eq:ss2}
\begin{align}
t_{x} : & \,(\sigma_i^x,\sigma_i^y,\sigma_i^z)_{i=1}^N\longrightarrow (-\sigma_i^x,\sigma_i^y,\sigma_i^z)_{i=1}^N,\\
t_{y} : & \,(\sigma_i^x,\sigma_i^y,\sigma_i^z)_{i=1}^N\longrightarrow (\sigma_i^x,-\sigma_i^y,\sigma_i^z)_{i=1}^N,\\
t_{z} : & \,(\sigma_i^x,\sigma_i^y,\sigma_i^z)_{i=1}^N\longrightarrow (\sigma_i^x,\sigma_i^y,-\sigma_i^z)_{i=1}^N.
\end{align}
\end{subequations}
Notice that the replacement rules \eqref{rule1}--\eqref{rule6} are not invariant under any of the sign flips~\eqref{eq:ss2}. Nevertheless, these extra sign symmetries of the Hamiltonian can still help to reduce the SDP
size.
For a monomial $u\in\sP_N$, we define its \emph{signature} with respect to the sign flips~\eqref{eq:ss2} as 
\begin{equation}
    \eta(u)\coloneqq(t_{x}(u)/u, t_{y}(u)/u, t_{z}(u)/u) \in \{-1, 1\}^3.
\end{equation}
According to the parity of the monomial degree, each subbasis $\cB_d^{(i)}$ can be partitioned further into two groups: $\cB_d^{(i,1)}=\left\{u\in\cB_d^{(i)}\,\middle|\, 2\mid\deg(u)\right\}$ and $\cB_d^{(i,2)}=\left\{u\in\cB_d^{(i)}\,\middle|\, 2\nmid\deg(u)\right\}$. Here, $2\mid\deg(u)$ and $2\nmid\deg(u)$ mean that the degree of the monomial $u$ is even and odd, respectively. Monomials in each $\cB_d^{(i,j)}$ share the same signature with respect to the sign flips~\eqref{eq:ss2}. We arrange the monomials in each $\cB_d^{(i)}$ so that $\cB_d^{(i,1)}$ precedes $\cB_d^{(i,2)}$, i.e., $\cB_d^{(i)}=\left[\cB_d^{(i,1)},\cB_d^{(i,2)}\right]$.

We now prove a few lemmas for later use.
\begin{lemma}\label{lm0}
Let $i\in\{1,2,3,4\}$ and $j\in\{1,2\}$. For any $v,w\in\cB_d^{(i,j)}$, one has 
\begin{equation}
\sC(v^{\star}w)\in\begin{cases}
\{-1,1\},&\text{if }2\mid\deg(\sM(v^{\star}w)),\\
\{-\bi,\bi\},&\text{otherwise}.
\end{cases}
\end{equation}
For any $v\in\cB_d^{(i,1)}$ and $w\in\cB_d^{(i,2)}$, one has
\begin{equation}
\sC(v^{\star}w)\in\begin{cases}
\{-1,1\},&\text{if }2\nmid\deg(\sM(v^{\star}w)),\\
\{-\bi,\bi\},&\text{otherwise}.
\end{cases}
\end{equation}
\end{lemma}
\begin{proof}
Suppose that $v,w\in\cB_d^{(i,j)}$. Then $2\mid\deg(v^{\star}w)$. If $2\mid\deg(\sM(v^{\star}w))$, then the replacement rules \eqref{rule1}-\eqref{rule6} are performed even times to obtain the normal form of $v^{\star}w$, from which we deduce that $\sC(v^{\star}w)\in\{-1,1\}$. If $2\nmid\deg(\sM(v^{\star}w))$, then the replacement rules \eqref{rule1}-\eqref{rule6} are performed odd times to obtain the normal form of $v^{\star}w$, from which we deduce that $\sC(v^{\star}w)\in\{-\bi,\bi\}$.

Suppose that $v\in\cB_d^{(i,1)}$ and $w\in\cB_d^{(i,2)}$. Then $2\nmid\deg(v^{\star}w)$. If $2\nmid\deg(\sM(v^{\star}w))$, then the replacement rules \eqref{rule1}-\eqref{rule6} are performed even times to obtain the normal form of $v^{\star}w$, from which we deduce that $\sC(v^{\star}w)\in\{-1,1\}$. If $2\mid\deg(\sM(v^{\star}w))$, then the replacement rules \eqref{rule1}-\eqref{rule6} are performed odd times to obtain the normal form of $v^{\star}w$, from which we deduce that $\sC(v^{\star}w)\in\{-\bi,\bi\}$.
\end{proof}

\begin{lemma}\label{lm1}
Let $A,B\in\C^{n\times n}$ such that $A^*=A$ and $B^*=-B$. Then
\begin{equation}
A+B\bi\succeq0\iff\begin{bmatrix}
    A&B\\-B&A
\end{bmatrix}\succeq0.
\end{equation}
\end{lemma}
\begin{proof}
Let $\bu,\bv\in\C^n$ be arbitrary. Then,
\begin{align*}
A+B\bi\succeq0&\iff(\bu-\bv\cdot\bi)^*(A+B\bi)(\bu-\bv\cdot\bi)=\bu^*A\bu+\bv^*A\bv-\bv^*B\bu+\bu^*B\bv\ge0\\
&\iff[\bu^* \,\,\bv^*]\begin{bmatrix}A&B\\-B&A\end{bmatrix}\begin{bmatrix}\bu\\\bv\end{bmatrix}\ge0\\
&\iff\begin{bmatrix}A&B\\-B&A\end{bmatrix}\succeq0.
\end{align*}
\end{proof}

\begin{lemma}\label{lm2}
Let $A\in\C^{n\times n}, B\in\C^{m\times m}, C\in\C^{m\times n}$ such that $A^*=A$ and $B^*=B$. Then,
\begin{equation}
\begin{bmatrix}
A&C\bi\\
-C^{*}\bi&B
\end{bmatrix}\succeq0\iff
\begin{bmatrix}
A&C\\
C^{*}&B
\end{bmatrix}\succeq0.
\end{equation}
\end{lemma}
\begin{proof}
By Lemma \ref{lm1}, we have
\begin{align*}
\begin{bmatrix}
A&C\bi\\
-C^{*}\bi&B
\end{bmatrix}\succeq0&\iff
\begin{bmatrix}
A&0&0&C\\
0&B&-C^*&0\\
0&-C&A&0\\
C^*&0&0&B
\end{bmatrix}\succeq0\\
&\iff\begin{bmatrix}
A&C\\
C^*&B
\end{bmatrix}\succeq0,\quad\begin{bmatrix}
B&-C^*\\
-C&A
\end{bmatrix}\succeq0\\
&\iff\begin{bmatrix}
A&C\\
C^*&B
\end{bmatrix}\succeq0.
\end{align*}
\end{proof}

\begin{proposition}\label{prop2}
In the moment relaxation \eqref{mom} for Heisenberg models, there is no loss of generality in assuming that $\ell(u)=0$ whenever $\eta(u)\ne(1,1,1)$.
\end{proposition}
\begin{proof}
For $i\in\{1,2,3,4\}$, let $\MM_{d}^{(i)}(\ell)$ be the block of the moment matrix $\MM_{d}(\ell)$ indexed by the subbasis $\cB_d^{(i)}$. For each $i$, we define the matrices $A^{(i)},B^{(i)}\in\C^{|\cB_d^{(i,1)}|\times|\cB_d^{(i,1)}|},C^{(i)},D^{(i)}\in\C^{|\cB_d^{(i,1)}|\times|\cB_d^{(i,2)}|},E^{(i)},F^{(i)}\in\C^{|\cB_d^{(i,2)}|\times|\cB_d^{(i,2)}|}$ as follows:
\begin{subequations}\label{sec:eq2}
\begin{align}
    A^{(i)}_{v,w}&=\begin{cases}
    \sC(v^{\star}w)\ell(\sM(v^{\star}w)), \quad&\text{if }2\mid\deg(\sM(v^{\star}w)),\\
    0, \quad&\text{otherwise},
    \end{cases}\quad \text{where } v,w\in\cB_d^{(i,1)};\\
    B^{(i)}_{v,w}&=\begin{cases}
    -\bi\sC(v^{\star}w)\ell(\sM(v^{\star}w)), &\text{if }2\nmid\deg(\sM(v^{\star}w)),\\
    0, &\text{otherwise},
    \end{cases}\quad \text{where } v,w\in\cB_d^{(i,1)};\\
    C^{(i)}_{v,w}&=\begin{cases}
    \sC(v^{\star}w)\ell(\sM(v^{\star}w)), \quad&\text{if }2\nmid\deg(\sM(v^{\star}w)),\\
    0, \quad&\text{otherwise},
    \end{cases}\quad \text{where } v\in\cB_d^{(i,1)},w\in\cB_d^{(i,2)};\\
    D^{(i)}_{v,w}&=\begin{cases}
    -\bi\sC(v^{\star}w)\ell(\sM(v^{\star}w)), &\text{if }2\mid\deg(\sM(v^{\star}w)),\\
    0, &\text{otherwise},
    \end{cases}\quad \text{where } v\in\cB_d^{(i,1)},w\in\cB_d^{(i,2)};\\
    E^{(i)}_{v,w}&=\begin{cases}
    \sC(v^{\star}w)\ell(\sM(v^{\star}w)), \quad&\text{if }2\mid\deg(\sM(v^{\star}w)),\\
    0, \quad&\text{otherwise},\end{cases}\quad \text{where } v,w\in\cB_d^{(i,2)};
    \\
    F^{(i)}_{v,w}&=\begin{cases}
    -\bi\sC(v^{\star}w)\ell(\sM(v^{\star}w)), &\text{if }2\nmid\deg(\sM(v^{\star}w)),\\
    0, &\text{otherwise},
    \end{cases}\quad \text{where } v,w\in\cB_d^{(i,2)}.
    \end{align}
\end{subequations}
For brevity, we will omit the superscript $(i)$ of $A,B,C,D,E,F$ in the rest of this proof.
Then by construction, we have
\begin{equation}\label{sec4:eq1}
    \MM_{d}^{(i)}(\ell)=\begin{bmatrix}
        A+B\bi&C+D\bi\\
        C^{*}-D^{*}\bi&E+F\bi
    \end{bmatrix}.
\end{equation}
For any feasible solution $\ell$ to \eqref{mom}, we define a linear functional $\tilde{\ell}:[\sP_N]_{2d}\rightarrow\C$ by
\begin{equation}
    \tilde{\ell}(u)\coloneqq\begin{cases}\ell(u),&\text{if }\eta(u)=(1,1,1),\\
    0, &\text{otherwise}.\end{cases}
\end{equation}
We next show that $\tilde{\ell}$ is again a feasible solution to \eqref{mom}.
First note that for any $v,w\in\cB_d^{(i)}$, either $\eta(\sM(v^{\star}w))=(1,1,1)$ if $2\mid\deg(\sM(v^{\star}w))$ or $\eta(\sM(v^{\star}w))=(-1,-1,-1)$ if $2\nmid\deg(\sM(v^{\star}w))$, which gives that
\begin{equation}
\MM_{d}^{(i)}(\tilde{\ell})=\begin{bmatrix}
A&D\bi\\
-D^{*}\bi&E
\end{bmatrix}.
\end{equation}
The construction \eqref{sec:eq2} gives $A^*=A$, $E^*=E$, $B^*=-B$, and $F^*=-F$. Applying Lemma~\ref{lm1}, we conclude that \eqref{sec4:eq1} is equivalent to
\begin{equation}
\begin{bmatrix}
A&C&B&D\\
C^*&E&-D^*&F\\
-B&-D&A&C\\
D^*&-F&C^*&E
\end{bmatrix}\succeq0,
\end{equation}
which implies
\begin{equation}
\begin{bmatrix}
A&D\\
D^{*}&E
\end{bmatrix}\succeq0.
\end{equation}
Thus by Lemma \ref{lm2}, we have
\begin{equation}
\MM_{d}^{(i)}(\tilde{\ell})=\begin{bmatrix}
A&D\bi\\
-D^{*}\bi&E
\end{bmatrix}\succeq0.
\end{equation}
Therefore, $\tilde{\ell}$ is feasible for \eqref{mom}.
Moreover, the monomials involved in the Hamiltonian $H$ are all of even degree and we thus have $\tilde{\ell}(H)=\ell(H)$.
This completes the proof.
\end{proof}

From now on, we assume that $\ell(u)=0$ whenever $\eta(u)\ne(1,1,1)$ in the moment relaxation~\eqref{mom}.

\subsubsection{Conjugate symmetry of the Hamiltonian}
The Hamiltonians of the Heisenberg models are invariant under complex conjugation. Consequently, the SDPs can be formulated over the real numbers.
\begin{proposition}\label{prop3}
In the moment relaxation \eqref{mom} for Heisenberg models, there is no loss of generality in assuming that $\ell(u)\in\bR$ for $u\in\tilde{W}_{2d}$ with $\eta(u)=(1,1,1)$.
\end{proposition}
\begin{proof}
From the proof of Proposition \ref{prop2}, we see that the PSD constraints of \eqref{mom} can be expressed as
\begin{equation}\label{sec4:eq3}
\MM_{d}^{(i)}(\ell)=\begin{bmatrix}
A^{(i)}&D^{(i)}\bi\\
-(D^{(i)})^{*}\bi&E^{(i)}
\end{bmatrix}\succeq0, \quad i = 1,2,3,4.
\end{equation}
By Lemma \ref{lm2}, \eqref{sec4:eq3} is equivalent to
\begin{equation}\label{sec4:eq4}
\begin{bmatrix}
A^{(i)}&D^{(i)}\\
(D^{(i)})^{*}&E^{(i)}
\end{bmatrix}\succeq0, \quad i = 1,2,3,4.
\end{equation}
For any feasible solution $\ell$ to \eqref{mom}, we define a linear functional $\tilde{\ell}:[\sP_N]_{2d}\rightarrow\C$ by 
\begin{equation}
 \tilde{\ell}(u)\coloneqq\frac{1}{2}\left(\ell(u)+\overline{\ell(u)}\right), \quad\forall u\in\tilde{W}_{2d}.
\end{equation}
Clearly, $\tilde{\ell}$ satisfies that $\tilde{\ell}(u)\in\bR$ for $u\in\tilde{W}_{2d}$.
By Lemma \ref{lm0}, the PSD constraints \eqref{sec4:eq4} for $\tilde{\ell}$ become
\begin{equation}\label{sec4:eq5}
\begin{bmatrix}
\Re(A^{(i)})&\Re(D^{(i)})\\
\Re((D^{(i)})^{*})&\Re(E^{(i)})
\end{bmatrix}\succeq0, \quad i = 1,2,3,4,
\end{equation}
where we use $\Re(\cdot)$ to denote the real part of a matrix.
So $\tilde{\ell}$ is again a feasible solution to \eqref{mom}. Moreover, we have $\tilde{\ell}(H)=\ell(H)$. We thus complete the proof.
\end{proof}

\subsubsection{Translation symmetry of the lattices}\label{ts}
Because the Heisenberg Hamiltonians are invariant under lattice translations, we obtain the following moment replacement rule.
\begin{proposition}\label{prop4}
In the moment relaxation \eqref{mom} for Heisenberg models, there is no loss of generality in assuming that 
\begin{equation}\label{sec4:eq6}
    \ell(\upsilon_k(u))=\ell(u), \quad \forall u\in\tilde{W}_{2d},
\end{equation}
for any translation of sites $\upsilon_k\colon i\rightarrow i+k$, $k\in[L]$.
\end{proposition}
\begin{proof}
For any feasible solution $\ell$ to \eqref{mom}, consider the linear functional $\tilde{\ell}$ defined by
\begin{equation}
    \tilde{\ell}(u)=\frac{1}{L}\sum_{k=1}^L\ell(u), \quad \forall u\in\tilde{W}_{2d}.
\end{equation}
It is straightforward to verify that $\tilde{\ell}$ is feasible for \eqref{mom} and that $\tilde{\ell}(H)=\ell(H)$. This proves the claim.
\end{proof}

Under periodic boundary conditions, the relation \eqref{sec4:eq6} induces a block structure in each $\MM_d^{(i)}(\ell)$, $i=1,2,3,4$, after the corresponding monomial subbasis is ordered appropriately. Every block is then an $L\times L$ circulant matrix, as shown in~\cite{wang2024certifying}.
For instance, in the 1D case, consider the block $T$ indexed by $\{\sigma_i^x\}_{i=1}^L$. We have $T_{i,j}=T_{j,i}=\ell(\sigma_i^x\sigma_j^x)=\ell(\sigma_1^x\sigma_{j-i+1}^x)$ which implies that $T$ is a symmetric circulant matrix.
Any circulant matrix of size $L$ can be diagonalized by a discrete Fourier transform $P\in\C^{L\times L}$ with
\begin{equation}
    P_{i,j}=\frac{1}{\sqrt{L}}\mathrm{e}^{-2\pi\bi(i-1)(j-1)/L},\quad i,j=1,\ldots,L,
\end{equation}
and the resulting diagonal elements are
\begin{equation}\label{sec4:eq7}
    \lambda_k=c_0+c_1\omega^{-k}+c_2\omega^{-2k}+\cdots+c_{L-1}\omega^{-(L-1)k},\quad k=0,\ldots,L-1,
\end{equation}
where $\omega=\mathrm{e}^{\frac{2\pi\bi}{L}}$ is a primitive $L$-th root of unity.
This fact enables us to block-diagonalise $\MM_{d}^{(i)}(\ell), i=1,2,3,4$ as follows. 

Let $\cB_d^{(i)}$ be arranged such that monomials of the same type with varying site label $i$ appear contiguously for $i=1,\ldots,L$ (e.g., $\sigma_1^x\sigma_{2}^x,\sigma_2^x\sigma_{3}^x,\ldots,\sigma_L^x\sigma_{1}^x$). Then, $\mathbf{M}^{(1)}_{d}(\ell)$ has the following block form ($t\coloneqq(|\cB_d^{(1)}|-1)/L$):
\begin{equation}
    G\coloneqq\begin{bmatrix}
    1&\bc_1^{\intercal}&\bc_2^{\intercal}&\cdots&\bc_t^{\intercal}\\
    \bc_1&G_{1,1}&G_{1,2}&\cdots&G_{1,t}\\
    \bc_2&G_{2,1}&G_{2,2}&\cdots&G_{2,t}\\
    \vdots&\vdots&\vdots&\ddots&\vdots\\
    \bc_t&G_{t,1}&G_{t,2}&\cdots&G_{t,t}\\
    \end{bmatrix}\,,
\end{equation}
where $\bc_j\coloneqq(c_j,\ldots,c_j)\in\bR^L$ and each $G_{j,k}$ is a circulant matrix. Let $U_G=\diag(1,P,\ldots,P)\in\C^{(1+Lt)\times(1+Lt)}$. We have $G=U_GD_GU_G^{*}$ where 
\begin{equation}
    D_G\coloneqq\begin{bmatrix}
    1&\bd_1^{\intercal}&\bd_2^{\intercal}&\cdots&\bd_t^{\intercal}\\
    \bd_1&D_G^{1,1}&D_G^{1,2}&\cdots&D_G^{1,t}\\
    \bd_2&D_G^{2,1}&D_G^{2,2}&\cdots&D_G^{2,t}\\
    \vdots&\vdots&\vdots&\ddots&\vdots\\
    \bd_t&D_G^{t,1}&D_G^{t,2}&\cdots&D_G^{t,t}\\
    \end{bmatrix}\,,
\end{equation}
with $\bd_j\coloneqq(c_j\sqrt{L},0,\ldots,0)\in\bR^L$ and diagonal matrices $D_G^{j,k}$.
Moreover, for $i=2,3,4$, $\mathbf{M}^{(i)}_{d}(\ell)$ has the following block form ($s\coloneqq(|\cB_d^{(i)}|-1)/L$):
\begin{equation}
    H\coloneqq\begin{bmatrix}
    H_{1,1}&H_{1,2}&\cdots&H_{1,s}\\
    H_{2,1}&H_{2,2}&\cdots&H_{2,s}\\
    \vdots&\vdots&\ddots&\vdots\\
    H_{s,1}&H_{s,2}&\cdots&H_{s,s}\\
    \end{bmatrix}\,,
\end{equation}
where each $H_{j,k}$ is a circulant matrix. Let $U_H=\diag(P,\ldots,P)\in\C^{(Ls)\times(Ls)}$. We have $H=U_HD_HU_H^{*}$ where 
\begin{equation}
    D_H\coloneqq\begin{bmatrix}
    D_H^{1,1}&D_H^{1,2}&\cdots&D_H^{1,s}\\
    D_H^{2,1}&D_H^{2,2}&\cdots&D_H^{2,s}\\
    \vdots&\vdots&\ddots&\vdots\\
    D_H^{s,1}&D_H^{s,2}&\cdots&D_H^{s,s}\\
    \end{bmatrix}\,,
\end{equation}
with diagonal matrices $D_H^{j,k}$. After reordering the rows and columns, both $D_G$ and $D_H$ are block diagonal. Hence, PSDness of $G$ and $H$ can be imposed separately on each diagonal block of $D_G$ and $D_H$.

Proposition~\ref{prop3} provides a further reduction on the SDP size by removing redundant PSD constraints. Specifically, it permits us to replace the complex PSD constraints \eqref{sec4:eq3} with the real PSD constraints \eqref{sec4:eq4} in the moment relaxation \eqref{mom}. As a result, the diagonal elements \eqref{sec4:eq7} satisfy
\begin{equation}
    \lambda_0\in\bR,\lambda_1=\overline{\lambda}_{L-1},\ldots,\lambda_{\frac{L}{2}-1}=\overline{\lambda}_{\frac{L}{2}+1},\lambda_{\frac{L}{2}}\in\bR.
\end{equation}
Therefore, for each $i\in\{1,\ldots,\frac{L}{2}-1\}$, the diagonal block consisting of the $\lambda_i$'s is conjugate to the diagonal block consisting of the $\lambda_{L-i}$'s, and so it suffices to impose PSD constraints on the former only.

In the 2D case, we may perform two rounds of block-diagonalization, first along the horizontal direction and second along the vertical direction. For example, consider the block $T$ indexed by $\{\sigma_{(i,j)}^x\}_{i,j=1}^L$. We have $T_{(i_1,i_2),(j_1,j_2)}=T_{(j_1,j_2),(i_1,i_2)}=\ell(\sigma_{(i_1,j_1)}^x\sigma_{(i_2,j_2)}^x)=\ell(\sigma_{1,1}^x\sigma_{(j_1-i_1+1,j_2-i_2+1)}^x)$. For each pair $(j_1,j_2)\in[L]\times[L]$, we perform the block-diagonalization procedure as above for the block with rows indexed by $\{\sigma_{(i,j_1)}^x\}_{i=1}^L$ and columns indexed by $\{\sigma_{(i,j_1)}^x\}_{i=1}^L$. Consequently, we obtain a block-diagonal matrix where each diagonal block is again a circulant matrix of size $L$. Now we can apply the block-diagonalization procedure once again to each diagonal block so that the block size is further reduced.

\subsubsection{Permutation symmetry on $x,y,z$}\label{ps}
The fact that the Hamiltonian $H$ of Heisenberg models is invariant under any permutation of $x,y,z$ yields the following moment replacement rule.
\begin{proposition}\label{prop5}
In the moment relaxation \eqref{mom} for Heisenberg models, there is no loss of generality in assuming that 
\begin{equation}\label{perms}
\ell(\tau(u))=\ell(u), \quad \forall u\in\tilde{W}_{2d},
\end{equation}
for any permutation $\tau$ acting on $x,y,z$.
\end{proposition}
\begin{proof}
Let $\mathcal{S}$ be the symmetric group acting on $x,y,z$, and let 
\begin{equation}
    \Sigma_N^{\mathcal{S}}\coloneqq\{f\in\sP_N\mid f\text{ is a SOHS such that } \tau(\NF(f)) = \NF(f),\, \forall\tau\in\mathcal{S}\}.
\end{equation}
By Theorem \ref{thm1}, any nonnegative element of $\sP_N$ that is invariant under $\mathcal{S}$ lies in $\Sigma_N^{\mathcal{S}}$.
The program \eqref{ncpop-d} thus becomes
\begin{equation}\label{ncpop-d1}
\begin{aligned}\max\limits_{\lambda}&\quad \lambda\\
    \text{subject to:}&\quad H-\lambda\in\Sigma_N^{\mathcal{S}},
    \end{aligned}
\end{equation}
whose dual reads as
\begin{equation}\label{momsym}
\begin{aligned}
    \min\limits_{\ell} &\quad \ell(H)\\
    \text{subject to:} &\quad \ell(f) \geq 0, \quad \forall f \in \Sigma_N^{\mathcal{S}},\\
    &\quad\ell(u^{\star})=\overline{\ell(u)},\quad \forall u\in\tilde{W},\\
    & \quad \ell(1)=1.
\end{aligned}
\end{equation}
For any feasible solution $\ell$ to \eqref{momsym}, we define a linear functional $\tilde{\ell} : \sP_N \to \C$ by $\tilde{\ell}(f) = \frac{1}{6}\sum_{\tau\in\mathcal{S}}\ell(\tau(\NF(f)))$. 
It is clear that $\tilde{\ell}(H) = \ell(H)$ since $H$ is invariant under the action of $\mathcal{S}$. 
In addition, $\tilde{\ell}(1)=1$, $\tilde{\ell}(u^{\star})=\overline{\tilde{\ell}(u)}$ for all $u\in\tilde{W}$, and $\tilde{\ell}(f)\geq0$ for all $f\in\Sigma_N^{\mathcal{S}}$.
Hence, $\tilde{\ell}$ is feasible for \eqref{momsym} and yields the same objective value as $\ell$. Moreover, $\tilde{\ell}(\tau(u))=\tilde{\ell}(u)$ for all $u\in\tilde{W}$ and $\tau\in\mathcal{S}$.
Therefore, there is no loss of generality in assuming the relations in \eqref{perms}.
\end{proof}

Proposition~\ref{prop5} implies a further reduction on the SDP size. That is, the PSD constraints $\MM_{d}^{(i)}(\ell)\succeq0$, $i=2,3,4$ are identical due to \eqref{perms}. Thus, it suffices to retain one of them in the moment relaxation \eqref{mom}.

\subsubsection{Mirror/dihedral symmetry of the lattices}\label{ms}
For 1D Heisenberg models, invariance of the Hamiltonian under the action of the mirror symmetry of the chain lattice yields the following moment replacement rule.
\begin{proposition}\label{prop60}
In the moment relaxation \eqref{mom} for 1D Heisenberg models, there is no loss of generality in assuming that 
\begin{equation}
\ell(\omega(u))=\ell(u), \quad \forall u\in\tilde{W}_{2d},
\end{equation}
where $\omega$ maps the subscript $i$ to $L-i$. 
\end{proposition}
\begin{proof}
For any feasible solution $\ell$ to \eqref{mom}, consider the linear functional $\tilde{\ell}$ defined by
\begin{equation}
    \tilde{\ell}(u)=\frac{1}{2}\bigl(\ell(u)+\ell(\omega(u))\bigr), \quad \forall u\in\tilde{W}_{2d}.
\end{equation}
It is straightforward to verify that $\tilde{\ell}$ is feasible for \eqref{mom} and that $\tilde{\ell}(H)=\ell(H)$. This proves the claim.
\end{proof}

For 2D Heisenberg models, invariance of the Hamiltonian under the action of the symmetry group of the square lattice yields the following moment replacement rule.
\begin{proposition}\label{prop6}
Let $\mathcal{D}_4$ denote the symmetry group of the square lattice, i.e., the dihedral group. In the moment relaxation \eqref{mom} for a 2D Heisenberg model, there is no loss of generality in assuming that
\begin{equation}
\ell(\omega(u))=\ell(u), \quad \forall u\in\tilde{W}_{2d},
\end{equation}
for any $\omega\in\mathcal{D}_4$. 
\end{proposition}
\begin{proof}
For any feasible solution $\ell$ to \eqref{mom}, consider the linear functional $\tilde{\ell}$ defined by
\begin{equation}
    \tilde{\ell}(u)=\frac{1}{8}\sum_{\omega\in\mathcal{D}_4}\ell(\omega(u)), \quad \forall u\in\tilde{W}_{2d}.
\end{equation}
It is straightforward to verify that $\tilde{\ell}$ is feasible for \eqref{mom} and that $\tilde{\ell}(H)=\ell(H)$. This proves the claim.
\end{proof}

\vspace{1em}
To end this section, we list the maximal block sizes of different SDP relaxations for 1D Heisenberg models in Table~\ref{blocksize}, from which we could see that a remarkable reduction on the
SDP size is achieved by fully exploiting algebraic structures.

\begin{table}[htbp]\label{blocksize}
\caption{Maximal block sizes of different SDP relaxations for 1D Heisenberg models ($r=1$).}
\centering
\renewcommand\arraystretch{1.5}
\begin{tabular}{|c|c|c|}
\hline
&General $N,d$&$N=100,d=4$\\
\hline
The original SDP&$\frac{(3N)^{d+1}-1}{3N-1}$&$8,127,090,301$\\
\hline
After exploiting equalities&$\sum_{i=0}^d\binom{N}{i}\cdot3^i$&$322,029,976$\\
\hline
After exploiting sparsity&$\frac{3N\cdot(3^d-1)}{2}+1$&$12,001$\\
\hline
After exploiting symmetry&$\frac{3^{d+1}-1}{8}$ if $2\nmid d$; $\frac{3^{d+1}+5}{8}$ if $2\mid d$&$31$\\
\hline
\end{tabular}
\end{table}

\section{Strengthening the SDP relaxations}
In parallel with reducing the SDP size, we can strengthen the moment relaxation \eqref{mom} by incorporating additional valid constraints.
\subsection{Positivity constraints on reduced density matrices}
For any integer $k\ge1$, the $k$-body reduced density matrix is defined as
\begin{equation}\label{positivity}
    \rho_{[k]} = \frac{1}{2^k}\sum_{a_1,\ldots,a_k}\langle\sigma_1^{a_1}\sigma_2^{a_2}\cdots\sigma_k^{a_k}\rangle\pi(\sigma_1^{a_1}\sigma_2^{a_2}\cdots\sigma_k^{a_k}),
\end{equation}
where $a_i\in\{0,x,y,z\}$, $i=1,\ldots,k$ (with $\sigma_i^0 \coloneqq I_2$), and $\langle\cdot\rangle$ means the expectation value at the ground state. Note that $\rho_{[k]}$ is a real PSD matrix of size $2^k \times 2^k$.
Therefore, we may strengthen the moment relaxation \eqref{mom} by requiring that
\begin{equation}\label{positivity1}
\tilde{\rho}_{[k]}(\ell)\coloneqq\sum_{a_1,\ldots,a_k}\ell(\sigma_1^{a_1}\sigma_2^{a_2}\cdots\sigma_k^{a_k})\pi(\sigma_1^{a_1}\sigma_2^{a_2}\cdots\sigma_k^{a_k})\succeq0.
\end{equation}

Note also that the Hamiltonians of the Heisenberg models are U(1)-invariant along the $z$ axis, implying that the reduced density matrix $\rho_{[k]}$ commutes with the total magnetization $\sum_{i=1}^k\sigma_i^z$. Consequently, $\rho_{[k]}$ splits into magnetization blocks: $\rho_{[k]}=\bigoplus_m\rho_{[k]}^{(m)}$, where $m=-\tfrac{k}{2},-\tfrac{k}{2}+1,\dots,\tfrac{k}{2}$ is the subsystem magnetization and each block $\rho_{[k]}^{(m)}$ is of size $\binom{k}{\tfrac{k}{2}+m}$. Therefore, instead of imposing $\tilde{\rho}_{[k]}(\ell)\succeq0$, we impose $\tilde{\rho}^{(m)}_{[k]}(\ell)\succeq0$ blockwise. It is sufficient to retain the blocks with $m\ge0$, because the constraint for $m=-i$ is equivalent to that for $m=i$.

\subsection{The state optimality conditions}
The state optimality conditions can also be utilized to strengthen moment relaxations of noncommutative polynomial optimization problems \cite{araujo2023first,fawzi2024certified}. 

\begin{theorem}[\cite{araujo2023first}, Theorem 3.7]
Let $\ell$ be any optimal solution of \eqref{state}. Then the following state optimality conditions hold:
\begin{equation}\label{lso}
\ell([H,f])=0,\quad\forall f\in\sP_N,
\end{equation}
and
\begin{equation}\label{pso}
    \ell\left(fHf^*-\frac{1}{2}(Hff^*+ff^*H)\right)\ge0,\quad\forall f\in\sP_N.
\end{equation}
\end{theorem}

Let $\tilde{\cB}$ be any collection of monomials. The condition \eqref{lso} yields the following valid linear constraints for the moment relaxation \eqref{mom}:
\begin{equation}\label{lso1}
\ell([H,u])=0,\quad\forall u\in\tilde{\cB}.
\end{equation}
Moreover, the condition \eqref{pso} yields the following valid PSD constraint for the moment relaxation~\eqref{mom}:
\begin{equation}\label{pso1}
\left[\ell\left(vHw^*-\frac{1}{2}(Hvw^*+vw^*H)\right)\right]_{v,w\in\tilde{\cB}}\succeq0.
\end{equation}

\begin{remark}
The techniques developed in Section~\ref{structure} can also reduce the size of the PSD constraint \eqref{pso1}.
\end{remark}

\section{Numerical results}
In this section, we report numerical results for four Heisenberg models. The calculations were performed on a workstation with 32 cores and 1~TB of RAM, using {\tt MOSEK 11.0} as the SDP solver.
For the 1D case, the DMRG data were obtained with the Julia packages {\tt ITensors 0.9.14} and {\tt ITensorMPS 0.3.23}.
The code for reproducing the results is available at \url{https://github.com/wangjie212/QMBCertify}. 

\subsection{The 1D case}
The SDP lower bounds on the ground-state energy per spin for the Heisenberg chain \eqref{model1} are displayed in Table~\ref{Tab:B1-Energies}. For comparison, we also report the DMRG upper bounds and the SDP lower bounds obtained in~\cite{wang2024certifying}. We define the relative gap between an SDP lower bound (lb) and a DMRG upper bound (ub) as
\begin{equation*}
    \text{gap}\coloneqq\frac{\text{ub}-\text{lb}}{|\text{ub}|}\times100\%.
\end{equation*}
Table~\ref{Tab:B1-Energies} demonstrates that the new SDP relaxation remains highly accurate as the system size increases. For all tested chains with up to $N=100$ spins, the absolute difference between the new SDP lower bound and the DMRG upper bound is below $10^{-5}$, and the relative gap never exceeds $0.0019\%$. At $N=100$, the new relaxation reduces the gap from $0.0820\%$ to $0.0019\%$, an improvement by a factor of more than $40$. See also Figure~\ref{fig:1}.

\begin{table}[htbp]\label{Tab:B1-Energies}
\caption{Bounds on the ground-state energy per spin of the Heisenberg chain with $N$ spins.}
\centering
\renewcommand\arraystretch{1.2}
\[
\begin{tabular}{c|c|cc|cc}
\Xhline{1pt}
\multirow{2}{*}{$N$}&\multirow{2}{*}{DMRG}&\multicolumn{2}{c|}{SDP Old}&\multicolumn{2}{c}{SDP New}\\
&&value&gap&value&gap\\
\Xhline{1pt}
10 & -0.4515446 &-0.4515446&0.0000\%& -0.4515446 &0.0000\%\\
14 & -0.4473964& -0.4474032&0.0015\%&-0.4473964&0.0000\%\\
18 & -0.4457083& -0.4457344&0.0059\%&-0.4457085 & 0.0000\%\\
20 & -0.4452193& -0.4452516&0.0073\%&-0.4452196 & 0.0000\%\\
22 & -0.4448582& -0.4448981&0.0090\%&-0.4448585 & 0.0000\%\\
26 & -0.4443707& -0.4444334&0.0141\%&-0.4443714 & 0.0001\%\\
30 & -0.4440654& -0.4441512&0.0193\%&-0.4440668& 0.0003\%\\
34 & -0.4438616& -0.4439644&0.0232\%&-0.4438632 & 0.0004\%\\
38 & -0.4437189& -0.4438331&0.0257\%&-0.4437212 & 0.0005\%\\
40 & -0.4436630& -0.4437820&0.0268\%&-0.4436649 & 0.0005\%\\
42 & -0.4436150& -0.4437371&0.0275\%&-0.4436176 & 0.0006\%\\
46 & -0.4435370& -0.4436656&0.0290\%&-0.4435397 & 0.0006\%\\
50 & -0.4434771& -0.4436101&0.0300\%&-0.4434798 & 0.0006\%\\
60 & -0.4433762& -0.4435169&0.0317\%&-0.4433804 & 0.0009\%\\
80 & -0.4432758& -0.4435377&0.0591\%&-0.4432808 & 0.0011\%\\
100 & -0.4432295& -0.4435928&0.0820\%&-0.4432378 &0.0019\%\\
\Xhline{1pt}
\end{tabular}
\]
\end{table}

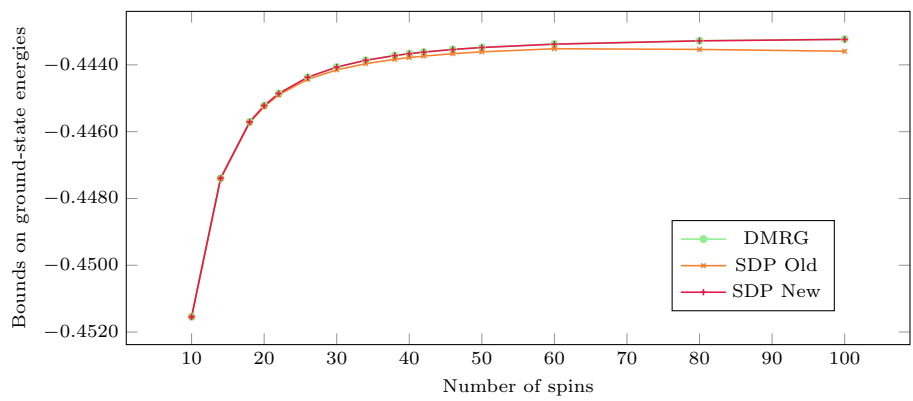
\begin{figure}
\centering
\begin{tikzpicture}
\footnotesize
\scalefont{0.8} 
\begin{axis}[
sharp plot, 
xmode=normal,
width=15cm, height=8cm,  
yticklabel ={\pgfmathprintnumber[fixed, fixed zerofill, precision=4]{\tick}},
xlabel = Number of spins,
ylabel = Bounds on ground-state energies,
xlabel near ticks,
ylabel near ticks,
legend style={at={(0.85,0.1)},anchor=south},
]

\addplot[semithick,mark=*,mark options={scale=0.6}, color=lightgreen] coordinates { 
     (10,     -0.4515446)
     (14,     -0.4473964)
     (18,     -0.4457083)
     (20,     -0.4452193)
     (22,     -0.4448582)
     (26,    -0.4443707)
     (30,    -0.4440654)
     (34,    -0.4438616)
     (38,    -0.4437189)
     (40,    -0.4436630)
     (42,     -0.4436150)
     (46,     -0.4435370)
     (50,     -0.4434771)
     (60,     -0.4433762)
     (80,     -0.4432758)
     (100,     -0.4432295)
  };
\addlegendentry{DMRG}

\addplot[semithick,mark=x,mark options={scale=0.6}, color=color2] coordinates { 
     (10,     -0.4515446)
     (14,    -0.4474032)
     (18,    -0.4457344)
     (20,    -0.4452516)
     (22,    -0.4448981)
     (26,    -0.4444334)
     (30,    -0.4441512)
     (34,    -0.4439644)
     (38,    -0.4438331)
     (40,    -0.4437820)
     (42,     -0.4437371)
     (46,     -0.4436656)
     (50,     -0.4436101)
     (60,     -0.4435169)
     (80,     -0.4435377)
     (100,     -0.4435928)
  };
\addlegendentry{SDP Old}

\addplot[semithick,mark=+,mark options={scale=0.6}, color=color3] coordinates { 
     (10,     -0.4515446)
     (14,    -0.4473964)
     (18,    -0.4457085)
     (20,    -0.4452196)
     (22,    -0.4448585)
     (26,    -0.4443714)
     (30,    -0.4440668)
     (34,    -0.4438632)
     (38,    -0.4437212)
     (40,    -0.4436649)
     (42,     -0.4436176)
     (46,     -0.4435397)
     (50,     -0.4434798)
     (60,     -0.4433804)
     (80,     -0.4432808)
     (100,     -0.4432378)
  };
\addlegendentry{SDP New}

\end{axis}
\end{tikzpicture}
\caption{Bounds on ground-state energies for the Heisenberg chain with $N$ spins.}
\label{fig:1}
\end{figure}  

We now investigate the $J_1$--$J_2$ Heisenberg chain \eqref{model2} with $N=40$ spins.
The SDP lower bounds on the ground-state energy per spin are displayed in Table~\ref{Tab:Heisenberg-2nd-Energies-L40}.\footnote{The value $J_2\simeq0.24117$ is included because it is the critical coupling separating the gapless Luttinger-liquid phase from the spontaneously dimerized phase in the thermodynamic limit. It therefore provides a physically important benchmark near the Berezinskii--Kosterlitz--Thouless transition.} For comparison, we also include the DMRG upper bounds and the SDP lower bounds obtained in \cite{wang2024certifying} in the table.
Table~\ref{Tab:Heisenberg-2nd-Energies-L40} shows that the new SDP lower bounds consistently improve upon those of~\cite{wang2024certifying}. For $J_2\leq0.4$, the relative gaps remain below $0.08\%$, including at the critical coupling $J_2\simeq0.24117$. At the Majumdar--Ghosh point $J_2=0.5$, the SDP relaxation exactly reproduces the ground-state energy per spin, $-0.375$. Although the bounds become less tight in the strongly frustrated regime, the new relative gaps remain below $1\%$ for all reported values; notably, at $J_2=1$, the gap decreases from $1.622\%$ to $0.737\%$. See also Figure~\ref{fig:2}.

\begin{table}[htb!]
\caption{Bounds on the ground-state energy per spin for the $J_1$--$J_2$ Heisenberg chain ($N=40$).} \label{Tab:Heisenberg-2nd-Energies-L40}
\centering
\begin{tabular}{c|c|cc|cc}
\Xhline{1pt}
\multirow{2}{*}{$J_2$}&\multirow{2}{*}{DMRG}&\multicolumn{2}{c|}{SDP Old}&\multicolumn{2}{c}{SDP New}\\
&&value&gap&value&gap\\
\Xhline{1pt}
        0.1 & $-0.4258079$ &-0.4258544&0.011\%& $-0.4258311$ &0.005\%\\
        0.2 & $-0.4089165$ &-0.4089265&0.002\%& $-0.4089219$ &0.001\%\\
        0.24117 & $-0.4023294$ &-0.4023393&0.002\%& $-0.4023344$ &0.001\%\\
        0.3 & $-0.3934175$ &-0.3934631&0.012\%& $-0.3934452$ &0.007\%\\
        0.4 & $-0.3805405$ &-0.3809208&0.100\%& $-0.3808371$ & 0.078\%\\
        0.5 & $-0.3750000$ & -0.3750000&0.000\%&$-0.3750000$ &0.000\%\\
        0.6 & $-0.3808099$ & -0.3816720&0.226\%&$-0.3813683$ &0.147\%\\
        0.7 & $-0.3971992$ & -0.3995204&0.584\%&$-0.3987825$ &0.399\%\\
        0.8 & $-0.4217288$ & -0.4261340&1.045\%&$-0.4246890$ &0.702\%\\
        0.9 & $-0.4520075$ &-0.4583928&1.413\%& $-0.4564972$ & 0.993\%\\
        1.0 & $-0.4865652$ &-0.4944564&1.622\%& $-0.4901516$ &0.737\%\\
        1.5 & $-0.6857166$ &-0.6957046&1.457\%& $-0.6919869$ &0.914\%\\
        2.0 & $-0.9024346$ &-0.9101027&0.850\%& $-0.9067115$ &0.474\%\\
\Xhline{1pt}
\end{tabular}
\end{table}

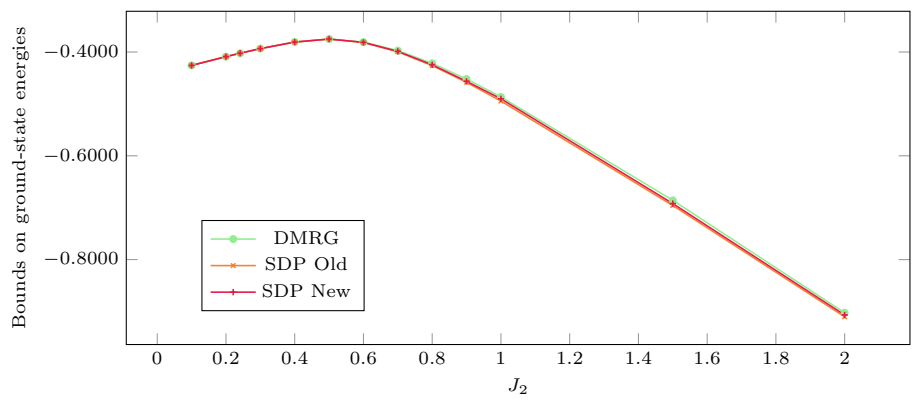
\begin{figure}
\centering
\subfigure[SDP lower
bounds and DMRG upper bounds]{
\begin{tikzpicture}
\footnotesize
\scalefont{0.8} 
\begin{axis}[
sharp plot, 
xmode=normal,
width=15cm, height=8cm,  
yticklabel ={\pgfmathprintnumber[fixed, fixed zerofill, precision=4]{\tick}},
xlabel = $J_2$,
ylabel = Bounds on ground-state energies,
xlabel near ticks,
ylabel near ticks,
legend style={at={(0.15,0.1)},anchor=south},
]

\addplot[semithick,mark=*,mark options={scale=0.6}, color=lightgreen] coordinates { 
     (0.1,     -0.4258079)
     (0.2,     -0.4089165)
     (0.24117, -0.4023294)
     (0.3,     -0.3934175)
     (0.4,     -0.380540)
     (0.5,     -0.3750000)
     (0.6,     -0.3808099)
     (0.7,     -0.3971992)
     (0.8,     -0.4217288)
     (0.9,     -0.4520075)
     (1.0,     -0.4865652)
     (1.5,     -0.6857166)
     (2.0,     -0.9024346)
  };
\addlegendentry{DMRG}

\addplot[semithick,mark=x,mark options={scale=0.6}, color=color2] coordinates { 
     (0.1,     -0.4258544)
     (0.2,     -0.4089265)
     (0.24117, -0.4023393)
     (0.3,     -0.3934631)
     (0.4,     -0.3809208)
     (0.5,     -0.3750000)
     (0.6,     -0.3816720)
     (0.7,     -0.3995204)
     (0.8,     -0.4261340)
     (0.9,     -0.4583928)
     (1.0,     -0.4944564)
     (1.5,     -0.6957046)
     (2.0,     -0.9101027)
  };
\addlegendentry{SDP Old}

\addplot[semithick,mark=+,mark options={scale=0.6}, color=color3] coordinates { 
     (0.1,     -0.4258311)
     (0.2,     -0.4089219)
     (0.24117, -0.4023344)
     (0.3,     -0.3934452)
     (0.4,     -0.3808371)
     (0.5,     -0.3750000)
     (0.6,     -0.3813683)
     (0.7,     -0.3987825)
     (0.8,     -0.4246890)
     (0.9,     -0.4564972)
     (1.0,     -0.4901516)
     (1.5,     -0.6919869)
     (2.0,     -0.9067115)
  };
\addlegendentry{SDP New}

\end{axis}
\end{tikzpicture}}\\
\subfigure[Differences from DMRG upper bounds]{
\begin{tikzpicture}
\footnotesize
\scalefont{0.8} 
\begin{axis}[
sharp plot, 
xmode=normal,
width=15cm, height=8cm,  
yticklabel ={\pgfmathprintnumber[fixed, fixed zerofill, precision=4]{\tick}},
xlabel = $J_2$,
ylabel = Differences,
xlabel near ticks,
ylabel near ticks,
legend style={at={(0.15,0.8)},anchor=south},
]

\addplot[semithick,mark=x,mark options={scale=0.6}, color=color2] coordinates 
{
(0.1, 0.0000465)
(0.2, 0.0000100)
(0.24117, 0.0000099)
(0.3, 0.0000456)
(0.4, 0.0003808)
(0.5, 0.0000000)
(0.6, 0.0008621)
(0.7, 0.0023212)
(0.8, 0.0044052)
(0.9, 0.0063853)
(1.0, 0.0078912)
(1.5, 0.0099880)
(2.0, 0.0076681)
};
\addlegendentry{SDP Old}

\addplot[semithick,mark=+,mark options={scale=0.6}, color=color3] coordinates 
{
(0.1, 0.0000232)
(0.2, 0.0000054)
(0.24117, 0.0000050)
(0.3, 0.0000277)
(0.4, 0.0002971)
(0.5, 0.0000000)
(0.6, 0.0005584)
(0.7, 0.0015833)
(0.8, 0.0029602)
(0.9, 0.0044897)
(1.0, 0.0035864)
(1.5, 0.0062703)
(2.0, 0.0042769)
};
\addlegendentry{SDP New}

\end{axis}
\end{tikzpicture}}
\caption{Bounds on ground-state energies for the $J_1$--$J_2$ Heisenberg chain ($N=40$).}
\label{fig:2}
\end{figure}  

The SDP lower and upper bounds on the first-neighbor correlation $C(1)=\frac{1}{4}\langle \sigma_1^x \sigma_{2}^x \rangle$ are displayed in Table~\ref{tab:B2-C1}, together with the DMRG values and the widths of certified SDP intervals.
The table shows that the new SDP relaxation substantially tightens the bounds on the first-neighbour correlation while consistently enclosing the DMRG values. For $J_2\leq0.3$, the new intervals are extremely narrow, with widths at most $2\times10^{-4}$. In the strongly frustrated regime $J_2\geq0.6$, their widths are reduced by approximately $37\%$--$52\%$ relative to the previous bounds. At the Majumdar--Ghosh point $J_2=0.5$, however, the old and new intervals have the same width. See also Figure~\ref{fig:3}.

\begin{table}[htb!]
\caption{Bounds on $C(1)$ for the $J_1$--$J_2$ Heisenberg chain ($N=40$).}
\label{tab:B2-C1}
\centering
\resizebox{\linewidth}{!}{
\begin{tabular}{c|c|c|c||c|c}
\Xhline{1pt}
$J_2$ & SDP Lower Bound & DMRG & SDP Upper Bound &Width Old&Width New\\
\Xhline{1pt}
0.1 & $-0.1477549$ & $-0.1477430$ & $-0.1477388$ &0.0002& 0.0000 \\
0.2 & $-0.1471812$ & $-0.1471695$ & $-0.1471620$ &0.0002& 0.0000 \\
0.241167 & $-0.1467381$ & $-0.1467200$ & $-0.1467050$ &0.0003& 0.0000 \\
0.3 & $-0.1457520$ & $-0.1456777$ & $-0.1455677$ &0.0010& 0.0002 \\
0.4 & $-0.1419388$ & $-0.1407434$ & $-0.1392751$ &0.0058& 0.0027 \\
0.5 & $-0.1258659$ & $-0.1250000$ & $-0.1241231$ &0.0017& 0.0017 \\
0.6 & $-0.1075884$ & $-0.1065750$ & $-0.0990974$ &0.0178& 0.0085 \\
0.7 & $-0.0894464$ & $-0.0837011$ & $-0.0738499$ &0.0296& 0.0156 \\
0.8 & $-0.0750487$ & $-0.0662572$ & $-0.0536994$ &0.0370& 0.0213 \\
0.9 & $-0.0633195$ & $-0.0539502$ & $-0.0386686$ &0.0399& 0.0247 \\
1.0 & $-0.0536750$ & $-0.0413538$ & $-0.0289607$ &0.0395& 0.0247 \\
1.5 & $-0.0295407$ & $-0.0151518$ & $-0.0119358$ &0.0292& 0.0176 \\
2.0 & $-0.0168607$ & $-0.0091622$ & $-0.0073162$ &0.0198& 0.0095 \\
\Xhline{1pt}
\end{tabular}}
\end{table}

\begin{figure}
\centering
\begin{tikzpicture}
\footnotesize
\scalefont{0.8} 
\begin{axis}[
sharp plot, 
xmode=normal,
width=15cm, height=8cm,  
yticklabel ={\pgfmathprintnumber[fixed, fixed zerofill, precision=4]{\tick}},
xlabel = $J_2$,
ylabel = Bounds on $C(1)$,
xlabel near ticks,
ylabel near ticks,
legend style={at={(0.8,0.05)},anchor=south},
]

\addplot[semithick,mark=*,mark options={scale=0.6}, color=lightgreen] coordinates { 
     (0.1,     -0.1477430)
     (0.2,     -0.1471695)
     (0.24117, -0.1467200)
     (0.3,     -0.1456777)
     (0.4,     -0.1407434)
     (0.5,     -0.1250000)
     (0.6,     -0.1065750)
     (0.7,     -0.0837011)
     (0.8,     -0.0662572)
     (0.9,     -0.0539502)
     (1.0,     -0.0413538)
     (1.5,     -0.0151518)
     (2.0,     -0.0091622)
  };
\addlegendentry{DMRG}

\addplot[semithick,mark=x,mark options={scale=0.6}, color=color2] coordinates { 
     (0.1,     -0.1477549)
     (0.2,     -0.1471812)
     (0.24117, -0.1467381)
     (0.3,     -0.1457520)
     (0.4,     -0.1419388)
     (0.5,     -0.1258659)
     (0.6,     -0.1075884)
     (0.7,     -0.0894464)
     (0.8,     -0.0750487)
     (0.9,     -0.0633195)
     (1.0,     -0.0536750)
     (1.5,     -0.0295407)
     (2.0,     -0.0168607)
  };
\addlegendentry{SDP Lower Bound (New)}

\addplot[semithick,mark=x,mark options={scale=0.6}, color=color1] coordinates { 
     (0.1,     -0.1478607)
     (0.2,     -0.1472921)
     (0.24117, -0.1468619)
     (0.3,     -0.1460336)
     (0.4,     -0.1430408)
     (0.5,     -0.1258632)
     (0.6,     -0.1110265)
     (0.7,     -0.0947528)
     (0.8,     -0.0812174)
     (0.9,     -0.0704769)
     (1.0,     -0.0623007)
     (1.5,     -0.0368734)
     (2.0,     -0.0245622)
  };
\addlegendentry{SDP Lower Bound (Old)}

\addplot[semithick,mark=+,mark options={scale=0.6}, color=color3] coordinates { 
     (0.1,     -0.1477388)
     (0.2,     -0.1471620)
     (0.24117, -0.1467050)
     (0.3,     -0.1455677)
     (0.4,     -0.1392751)
     (0.5,     -0.1241231)
     (0.6,     -0.0990974)
     (0.7,     -0.0738499)
     (0.8,     -0.0536994)
     (0.9,     -0.0386686)
     (1.0,     -0.0289607)
     (1.5,     -0.0119358)
     (2.0,     -0.0073162)
  };
\addlegendentry{SDP Upper Bound (New)}

\addplot[semithick,mark=+,mark options={scale=0.6}, color=bordeaux] coordinates { 
     (0.1,     -0.1476572)
     (0.2,     -0.1470603)
     (0.24117, -0.1465536)
     (0.3,     -0.1450314)
     (0.4,     -0.1371950)
     (0.5,     -0.1241231)
     (0.6,     -0.0932175)
     (0.7,     -0.0651186)
     (0.8,     -0.0442571)
     (0.9,     -0.0305553)
     (1.0,     -0.0228391)
     (1.5,     -0.0076266)
     (2.0,     -0.0047742)
  };
\addlegendentry{SDP Upper Bound (Old)}

\end{axis}
\end{tikzpicture}
\caption{Bounds on $C(1)$ for the $J_1$--$J_2$ Heisenberg chain ($N=40$).}
\label{fig:3}
\end{figure}  

The SDP lower and upper bounds on the second-neighbour correlation $C(2)=\frac{1}{4}\langle \sigma_1^x \sigma_{3}^x \rangle$ are displayed in Table~\ref{tab:B2-C2}. For comparison, we also include the DMRG values and the widths of certified SDP intervals in the table.
Table~\ref{tab:B2-C2} shows that the new SDP relaxation substantially tightens the bounds on the second-neighbour correlation while consistently enclosing the DMRG values. For $J_2\leq0.3$, the new intervals are extremely narrow, with widths at most $6\times10^{-4}$. For $J_2\geq0.6$, their widths are reduced by approximately $38\%$--$55\%$ relative to the previous bounds. The sign change near the Majumdar--Ghosh point $J_2=0.5$ reflects the transition from positive to antiferromagnetic second-neighbour correlations; at this point, the exact DMRG value is zero and the old and new intervals have the same width. See also Figure~\ref{fig:4}.

\begin{table}[htb!]
\caption{Bounds on $C(2)$ for the $J_1$--$J_2$ Heisenberg chain ($N=40$).}
\label{tab:B2-C2}
\centering
\resizebox{\linewidth}{!}{
\begin{tabular}{c|c|c|c||c|c}
\Xhline{1pt}
$J_2$ & SDP Lower Bound & DMRG & SDP Upper Bound &Width Old&Width New\\
\Xhline{1pt}
0.1 & 0.0580278 & 0.0580709 & 0.0581402 &0.0020& 0.0001 \\
0.2 & 0.0542820 & 0.0543201 & 0.0543795 &0.0012& 0.0001 \\
0.241167 & 0.0522255 & 0.0522883 & 0.0523640 &0.0013& 0.0001 \\
0.3 & 0.0480941 & 0.0484617 & 0.0487100 &0.0033& 0.0006 \\
0.4 & 0.0310703 & 0.0347344 & 0.0377300 &0.0146& 0.0067 \\
0.5 & -0.0017568 & 0 & 0.0017330 &0.0035& 0.0035 \\
0.6 & -0.0463989 & -0.0381663 & -0.0322471 &0.0297& 0.0142 \\
0.7 & -0.0836426 & -0.0697204 & -0.0613618 &0.0423& 0.0223 \\
0.8 & -0.1085960 & -0.0929827 & -0.0819094 &0.0462& 0.0267 \\
0.9 & -0.1244451 & -0.1075861 & -0.0970552 &0.0444& 0.0274 \\
1.0 & -0.1332228 & -0.1208368 & -0.1085337 &0.0395& 0.0247 \\
1.5 & -0.1442046 & -0.1422496 & -0.1329467 &0.0195& 0.0113 \\
2.0 & -0.1465414 & -0.1457875 & -0.1420041 &0.0099& 0.0045 \\
\Xhline{1pt}
\end{tabular}}
\end{table}

\begin{figure}	
\centering
\begin{tikzpicture}
\footnotesize
\scalefont{0.8} 
\begin{axis}[
sharp plot, 
xmode=normal,
width=15cm, height=8cm,  
yticklabel ={\pgfmathprintnumber[fixed, fixed zerofill, precision=4]{\tick}},
xlabel = $J_2$,
ylabel = Bounds on $C(2)$,
xlabel near ticks,
ylabel near ticks,
legend style={at={(0.8,0.65)},anchor=south},
]

\addplot[semithick,mark=*,mark options={scale=0.6}, color=lightgreen] coordinates { 
     (0.1,     0.0580709)
     (0.2,     0.0543201)
     (0.24117, 0.0522883)
     (0.3,     0.0484617)
     (0.4,     0.0347344)
     (0.5,     0)
     (0.6,     -0.0381663)
     (0.7,     -0.0697204)
     (0.8,     -0.0929827)
     (0.9,     -0.1075861)
     (1.0,     -0.1208368)
     (1.5,     -0.1422496)
     (2.0,     -0.1457875)
  };
\addlegendentry{DMRG}

\addplot[semithick,mark=x,mark options={scale=0.6}, color=color2] coordinates { 
     (0.1,     0.0580278)
     (0.2,     0.0542820)
     (0.24117, 0.0522255)
     (0.3,     0.0480941)
     (0.4,     0.0310703)
     (0.5,     -0.0017568)
     (0.6,     -0.0463989)
     (0.7,     -0.0836426)
     (0.8,     -0.1085960)
     (0.9,     -0.1244451)
     (1.0,     -0.1332228)
     (1.5,     -0.1442046)
     (2.0,     -0.1465414)
  };
\addlegendentry{SDP Lower Bound (New)}

\addplot[semithick,mark=x,mark options={scale=0.6}, color=color1] coordinates { 
     (0.1,     0.0572121)
     (0.2,     0.0537717)
     (0.24117, 0.0515969)
     (0.3,     0.0463033)
     (0.4,     0.0258637)
     (0.5,     -0.0017547)
     (0.6,     -0.0561980)
     (0.7,     -0.0961199)
     (0.8,     -0.1204157)
     (0.9,     -0.1334790)
     (1.0,     -0.1393519)
     (1.5,     -0.1472941)
     (2.0,     -0.1480146)
  };
\addlegendentry{SDP Lower Bound (Old)}

\addplot[semithick,mark=+,mark options={scale=0.6}, color=color3] coordinates { 
     (0.1,     0.0581402)
     (0.2,     0.0543795)
     (0.24117, 0.0523640)
     (0.3,     0.0487100)
     (0.4,     0.0377300)
     (0.5,     0.0017330)
     (0.6,     -0.0322471)
     (0.7,     -0.0613618)
     (0.8,     -0.0819094)
     (0.9,     -0.0970552)
     (1.0,     -0.1085337)
     (1.5,     -0.1329467)
     (2.0,     -0.1420041)
  };
\addlegendentry{SDP Upper Bound (New)}

\addplot[semithick,mark=+,mark options={scale=0.6}, color=bordeaux] coordinates { 
     (0.1,     0.0592576)
     (0.2,     0.0549354)
     (0.24117, 0.0528797)
     (0.3,     0.0496505)
     (0.4,     0.0404804)
     (0.5,     0.0017307)
     (0.6,     -0.0265157)
     (0.7,     -0.0537843)
     (0.8,     -0.0742153)
     (0.9,     -0.0891219)
     (1.0,     -0.0998902)
     (1.5,     -0.1277965)
     (2.0,     -0.1381216)
  };
\addlegendentry{SDP Upper Bound (Old)}

\end{axis}
\end{tikzpicture}
\caption{Bounds on $C(2)$ for the $J_1$--$J_2$ Heisenberg chain ($N=40$).}
\label{fig:4}
\end{figure}  

The SDP lower and upper bounds on the spin-spin correlation $C(i)=\frac{1}{4}\langle \sigma_1^x \sigma_{i+1}^x \rangle$ with $J_2=0.2$ are reported in Table~\ref{tab:B2-spinspin-J202}. For comparison, we also include the DMRG values and the widths of certified SDP intervals in the table.
Table~\ref{tab:B2-spinspin-J202} shows that the new SDP intervals consistently contain the DMRG values and are substantially narrower than the previous bounds at every distance. The improvement is particularly pronounced at short distances: for $i\leq3$, the interval widths are reduced by more than $90\%$. Although the bounds gradually widen with increasing distance, the width at $i=20$ is still reduced from $0.0117$ to $0.0075$, or by approximately $36\%$. The results also accurately reproduce the alternating signs and decay of the antiferromagnetic spin correlations. See also Figure~\ref{fig:5}.

\begin{table}[htb!]
\caption{Bounds on the spin-spin correlation at distance $i$ for the $J_1$--$J_2$ Heisenberg chain ($J_2=0.2$ and $N=40$).}
\label{tab:B2-spinspin-J202}
\centering
\begin{tabular}{c|c|c|c||c|c}
\Xhline{1pt}
$i$ & SDP Lower Bound & DMRG& SDP Upper Bound &Width Old&Width New\\
\Xhline{1pt}
1 & -0.1471812 & -0.1471695 & -0.1471620 &0.0002& 0.0000 \\
2 & 0.0542820 & 0.0543201 & 0.0543794 &0.0012& 0.0001 \\
3 & -0.0415919 & -0.0415068 & -0.0414502 &0.0015& 0.0001 \\
4 & 0.0277486 & 0.0278817 & 0.0280607 &0.0018& 0.0003 \\
5 & -0.0253522 & -0.0250681 & -0.0248215 &0.0022& 0.0005 \\
6 & 0.0191478 & 0.0196057 & 0.0200717 &0.0027& 0.0009 \\
7 & -0.0190638 & -0.0183740 & -0.0176101 &0.0035& 0.0015 \\
8 & 0.0142895 & 0.0154256 & 0.0164151 &0.0042& 0.0021 \\
9 & -0.0160573 & -0.0148077 & -0.0133178 &0.0051& 0.0027 \\
10 & 0.0111842 & 0.0129733 & 0.0145166 &0.0058& 0.0033 \\
11 & -0.0144800 & -0.0126691 & -0.0105946 &0.0066& 0.0039 \\
12 & 0.0091031 & 0.0114280 & 0.0134886 &0.0072& 0.0044 \\
13 & -0.0136166 & -0.0113155 & -0.0087499 &0.0079& 0.0049 \\
14 & 0.0076319 & 0.0104312 & 0.0129356 &0.0085& 0.0053 \\
15 & -0.0131750 & -0.0104550 & -0.0074354 &0.0091& 0.0057 \\
16 & 0.0065791 & 0.0098061 & 0.0127040 &0.0097& 0.0061 \\
17 & -0.0130599 & -0.0099413 & -0.0065214 &0.0103& 0.0065 \\
18 & 0.0058487 & 0.0094607 & 0.0127624 &0.0108& 0.0069 \\
19 & -0.0132352 & -0.0097003 & -0.0059019 &0.0113& 0.0073 \\
20 & 0.0054422 & 0.0093500 & 0.0129749 &0.0117& 0.0075 \\
\Xhline{1pt}
\end{tabular}
\end{table}

\begin{figure}	
\centering
\begin{tikzpicture}
\footnotesize
\scalefont{0.8} 
\begin{axis}[
sharp plot, 
xmode=normal,
width=15cm, height=10cm,  
yticklabel ={\pgfmathprintnumber[fixed, fixed zerofill, precision=4]{\tick}},
xlabel = $i$,
ylabel = Bounds on $C(i)$,
xlabel near ticks,
ylabel near ticks,
legend style={at={(0.8,0.05)},anchor=south},
]

\addplot[semithick,mark=*,mark options={scale=0.6}, color=lightgreen] coordinates { 
     (1,     -0.1471695)
     (2,     0.0543201)
     (3,     -0.0415068)
     (4,     0.0278817)
     (5,     -0.0250681)
     (6,     0.0196057)
     (7,     -0.0183740)
     (8,     0.0154256)
     (9,     -0.0148077)
     (10,     0.0129733)
     (11,     -0.0126691)
     (12,     0.0114280)
     (13,     -0.0113155)
     (14,     0.0104312)
     (15,     -0.010455)
     (16,     0.0098061)
     (17,     -0.0099413)
     (18,     0.00946065)
     (19,     -0.0097003)
     (20,     0.0093500)
  };
\addlegendentry{DMRG}

\addplot[semithick,mark=x,mark options={scale=0.6}, color=color2] coordinates { 
     (1,     -0.1471812)
     (2,     0.0542820)
     (3,     -0.0415919)
     (4,     0.0277486)
     (5,     -0.0253522)
     (6,     0.0191478)
     (7,     -0.0190638)
     (8,     0.0142895)
     (9,     -0.0160573)
     (10,     0.0111842)
     (11,     -0.0144800)
     (12,     0.0091031)
     (13,     -0.0136166)
     (14,     0.0076319)
     (15,     -0.0131750)
     (16,     0.0065791)
     (17,     -0.0130599)
     (18,     0.0058487)
     (19,     -0.0132352)
     (20,     0.0054422)
  };
\addlegendentry{SDP Lower Bound (New)}

\addplot[semithick,mark=x,mark options={scale=0.6}, color=color1] coordinates { 
     (1,     -0.1472921)
     (2,     0.0537717)
     (3,     -0.0422374)
     (4,     0.0269796)
     (5,     -0.0260822)
     (6,     0.0182804)
     (7,     -0.0201472)
     (8,     0.0133911)
     (9,     -0.0174870)
     (10,     0.0101486)
     (11,     -0.0161036)
     (12,     0.0079224)
     (13,     -0.0154071)
     (14,     0.0062784)
     (15,     -0.0151364)
     (16,     0.0050437)
     (17,     -0.0151801)
     (18,     0.0041464)
     (19,     -0.0154778)
     (20,     0.0035877)
  };
\addlegendentry{SDP Lower Bound (Old)}

\addplot[semithick,mark=+,mark options={scale=0.6}, color=color3] coordinates { 
     (1,     -0.1471695)
     (2,     0.0543201)
     (3,     -0.0415068)
     (4,     0.0278817)
     (5,     -0.0250681)
     (6,     0.0196057)
     (7,     -0.0183740)
     (8,     0.0154256)
     (9,     -0.0148077)
     (10,     0.0129733)
     (11,     -0.0126691)
     (12,     0.0114280)
     (13,     -0.0113155)
     (14,     0.0104312)
     (15,     -0.010455)
     (16,     0.0098061)
     (17,     -0.0099413)
     (18,     0.00946065)
     (19,     -0.0097003)
     (20,     0.0093500)
  };
\addlegendentry{SDP Upper Bound (New)}

\addplot[semithick,mark=+,mark options={scale=0.6}, color=bordeaux] coordinates { 
     (1,     -0.1470603)
     (2,      0.0549354)
     (3,     -0.0407166)
     (4,     0.0287525)
     (5,     -0.0239108)
     (6,     0.0209771)
     (7,     -0.0166783)
     (8,     0.0176028)
     (9,     -0.0123538)
     (10,     0.0159793)
     (11,     -0.0094935)
     (12,     0.0151389)
     (13,     -0.0074851)
     (14,     0.0147637)
     (15,     -0.0059942)
     (16,     0.0147138)
     (17,     -0.0049027)
     (18,     0.0149118)
     (19,     -0.0041307)
     (20,     0.0152579)
  };
\addlegendentry{SDP Upper Bound (Old)}

\end{axis}
\end{tikzpicture}
\caption{Bounds on the spin-spin correlation at distance $i$ for the $J_1$--$J_2$ Heisenberg chain ($J_2=0.2$ and $N=40$).}
\label{fig:5}
\end{figure}  

Table~\ref{tab:B2-spinspin-J21} reports the SDP lower and upper bounds on the spin-spin correlation in the strongly frustrated regime $J_2=1$. It shows that every DMRG value lies within the corresponding certified SDP interval. Compared with the previous bounds in~\cite{wang2024certifying}, the new bounds are uniformly tighter over all distances, reducing the average interval width by approximately $25\%$. As expected from the predominantly local monomial basis, the certified intervals generally widen with distance, leaving room for further improvement through bases tailored to long-range correlations. See also Figure~\ref{fig:6}.

\begin{table}[htb!]
\caption{Bounds on the spin-spin correlation at distance $i$ for the Heisenberg chain with second-neighbour couplings ($J_2=1.0$ and $N=40$).}
    \label{tab:B2-spinspin-J21}
    \centering
    \begin{tabular}{c|c|c|c||c|c}
\Xhline{1pt}
$i$ & SDP Lower Bound &DMRG& SDP Upper Bound &Width Old&Width New\\
\Xhline{1pt}
1 & -0.0538357 & -0.0413538 & -0.0287978 &0.0395& 0.0250\\
2 & -0.1333906 & -0.1208368 & -0.1083527 &0.0395& 0.0250\\
3 & 0.0236042 & 0.0331171 & 0.0450558 &0.0354& 0.0215\\
4 & 0.0217878 & 0.0362300 & 0.0492959  &0.0478& 0.0275\\
5 & -0.0418949 & -0.0267034 & -0.0148445 &0.0535& 0.0271\\
6 & -0.0419023 & -0.0195505 & 0.0034863 &0.0700& 0.0454\\
7 & 0.0055856 & 0.0186073 & 0.0397842 &0.0556& 0.0342\\
8 & -0.0263595 & 0.0062481 & 0.0324218 &0.0810& 0.0588\\
9 & -0.0432032 & -0.0132155 & 0.0092146 &0.0779& 0.0524\\
10 & -0.0348900 & -0.0019804 & 0.0336909 &0.0876& 0.0686\\
11 & -0.0229635 & 0.0090596 & 0.0418082 &0.0831& 0.0648\\
12 & -0.0375145 & -0.0017641 & 0.0325153 & 0.0861 &0.0700\\
13 & -0.0431658 & -0.0059637 & 0.0315920 &0.0919& 0.0748\\
14 & -0.0360265 & 0.0027244 & 0.0335339 &0.0835& 0.0696\\
15 & -0.0380783 & 0.0038911 & 0.0364540 &0.0914& 0.0745\\
16 & -0.0426117 & -0.0039891 & 0.0346357 &0.0910& 0.0772\\
17 & -0.0334526 & -0.0020209 & 0.0309240 &0.0812& 0.0644\\
18 & -0.0437272 & 0.0039570 & 0.0428580 &0.1040& 0.0866\\
19 & -0.0180974 & 0.0007610 & 0.0162195 &0.0446& 0.0343\\
20 & -0.0545757 & -0.0044521 & 0.0449810 &0.1238& 0.0996\\
\Xhline{1pt}
    \end{tabular}
\end{table}

\begin{figure}	
\centering
\begin{tikzpicture}
\footnotesize
\scalefont{0.8} 
\begin{axis}[
sharp plot, 
xmode=normal,
width=15cm, height=10cm,  
yticklabel ={\pgfmathprintnumber[fixed, fixed zerofill, precision=4]{\tick}},
xlabel = $i$,
ylabel = Bounds on $C(i)$,
xlabel near ticks,
ylabel near ticks,
legend style={at={(0.8,0.05)},anchor=south},
]

\addplot[semithick,mark=*,mark options={scale=0.6}, color=lightgreen] coordinates { 
     (1,     -0.0413538)
     (2,     -0.1208368)
     (3,     0.0331171)
     (4,     0.0362300)
     (5,     -0.0267034)
     (6,     -0.0195505)
     (7,     0.0186073)
     (8,     0.0062481)
     (9,     -0.0132155)
     (10,     -0.0019804)
     (11,     0.0090596)
     (12,     -0.0017641)
     (13,     -0.0059637)
     (14,     0.0027244)
     (15,     0.0038911)
     (16,     -0.0039891)
     (17,     -0.0020209)
     (18,     0.0039570)
     (19,     0.0007610)
     (20,     -0.0044521)
  };
\addlegendentry{DMRG}

\addplot[semithick,mark=x,mark options={scale=0.6}, color=color2] coordinates { 
     (1,     -0.0538357)
     (2,     -0.1333906)
     (3,     0.0236042)
     (4,     0.0217878)
     (5,     -0.0418949)
     (6,     -0.0419023)
     (7,     0.0055856)
     (8,     -0.0263595)
     (9,     -0.0432032)
     (10,     -0.0348900)
     (11,     -0.0229635)
     (12,     -0.0375145)
     (13,     -0.0431658)
     (14,     -0.0360265)
     (15,     -0.0380783)
     (16,     -0.0426117)
     (17,     -0.0334526)
     (18,     -0.0437272)
     (19,     -0.0180974)
     (20,     -0.0545757)
  };
\addlegendentry{SDP Lower Bound (New)}

\addplot[semithick,mark=x,mark options={scale=0.6}, color=color1] coordinates { 
     (1,     -0.06230067)
     (2,     -0.13935193)
     (3,     0.01884823)
     (4,     0.00971664)
     (5,     -0.06031966)
     (6,     -0.05338124)
     (7,     -0.00390555)
     (8,     -0.04111449)
     (9,     -0.05681356)
     (10,     -0.04328585)
     (11,     -0.03371714)
     (12,     -0.04668599)
     (13,     -0.05242524)
     (14,     -0.04496886)
     (15,     -0.04945007)
     (16,     -0.05042092)
     (17,     -0.04231650)
     (18,     -0.05358065)
     (19,     -0.02405042)
     (20,     -0.06885791)
  };
\addlegendentry{SDP Lower Bound (Old)}

\addplot[semithick,mark=+,mark options={scale=0.6}, color=color3] coordinates { 
     (1,     -0.0287978)
     (2,     -0.1083527)
     (3,     0.0450558)
     (4,     0.0492959)
     (5,     -0.0148445)
     (6,     0.0034863)
     (7,     0.0397842)
     (8,     0.0324218)
     (9,     0.0092146)
     (10,     0.0336909)
     (11,     0.0418082)
     (12,     0.0325153)
     (13,     0.0315920)
     (14,     0.0335339)
     (15,     0.0364540)
     (16,     0.0346357)
     (17,     0.0309240)
     (18,     0.0428580)
     (19,     0.0162195)
     (20,     0.0449810)
  };
\addlegendentry{SDP Upper Bound (New)}

\addplot[semithick,mark=+,mark options={scale=0.6}, color=bordeaux] coordinates { 
     (1,     -0.02283913)
     (2,     -0.09989024)
     (3,     0.05421466)
     (4,     0.05751094)
     (5,     -0.00682832)
     (6,     0.01661007)
     (7,     0.05169353)
     (8,     0.03992840)
     (9,     0.02103763)
     (10,     0.04430916)
     (11,     0.04936674)
     (12,     0.03943488)
     (13,     0.03949383)
     (14,     0.03852069)
     (15,     0.04193050)
     (16,     0.04059603)
     (17,     0.03886868)
     (18,     0.05037295)
     (19,     0.02054356)
     (20,     0.05499198)
  };
\addlegendentry{SDP Upper Bound (Old)}

\end{axis}
\end{tikzpicture}
\caption{Bounds on the spin-spin correlation at distance $i$ for the $J_1$--$J_2$ Heisenberg chain ($J_2=1.0$ and $N=40$).}
\label{fig:6}
\end{figure}  


Here, we bound the structure factor $S(\pi,\pi)=\frac{1}{4N^2}\sum_{i,j=1}^N\sum_{a\in\{x,y,z\}}\langle\sigma_{i}^a\sigma_{j}^a\rangle \mathrm{e}^{\bi\pi(i-j)}$ for the $J_1$--$J_2$ Heisenberg chain with $N=40$ spins and present the results in Table~\ref{tab:structure factor}. For comparison, we also include the DMRG values in the table. 
Table~\ref{tab:structure factor} shows that the DMRG values are consistently enclosed by the certified SDP intervals. The structure factor $S(\pi,\pi)$ decreases markedly as $J_2$ increases, reflecting the suppression of antiferromagnetic correlations at momentum $\pi$ by frustration. At the Majumdar--Ghosh point $J_2=0.5$, the particularly narrow interval accurately captures the exact DMRG value $0.0375$. In the large-$J_2$ regime, $S(\pi,\pi)$ approaches zero, consistent with the emergence of two weakly coupled antiferromagnetic subchains whose dominant correlations occur near momentum $\pi/2$. See also Figure~\ref{fig:10}.

\begin{table}[htbp]\label{tab:structure factor}
\caption{Bounds on the structure factor $S(\pi,\pi)$ for the $J_1$--$J_2$ Heisenberg chain ($N=40$).}    
    \centering
\renewcommand\arraystretch{1.2}
\[
\begin{tabular}{c|c|c|c}
\Xhline{1pt}
$J_2$ & SDP Lower Bound & DMRG & SDP Upper Bound\\
\Xhline{1pt}
0.0 & 0.0985998 & 0.1031831& 0.1064600\\
0.1 & 0.0908153 & 0.0973667 & 0.1042138\\
0.2 & 0.0856754 & 0.0903022 & 0.0942894\\
0.241167 & 0.0811375 & 0.0867427 & 0.0910491\\
0.3 & 0.0690184 & 0.0804070 & 0.0885589\\
0.4 & 0.0403764 & 0.0608781 & 0.0798086\\
0.5 & 0.0324779 & 0.0375000 & 0.0424656\\
0.6 & 0.0190150 & 0.0266019 & 0.0313481\\
0.7 & 0.0111088 & 0.0184667 & 0.0254093\\
0.8 & 0.0066581 & 0.0135322 & 0.0214150\\
0.9 & 0.0041313 & 0.0102386 & 0.0185697\\
1.0 & 0.0032505 & 0.0072953 & 0.0155411\\
1.5 & 0.0007164 & 0.0017755 & 0.0088262\\
2.0 & 0.0003421 & 0.0006739 & 0.0045430\\
\Xhline{1pt}
\end{tabular}
    \]
\end{table}

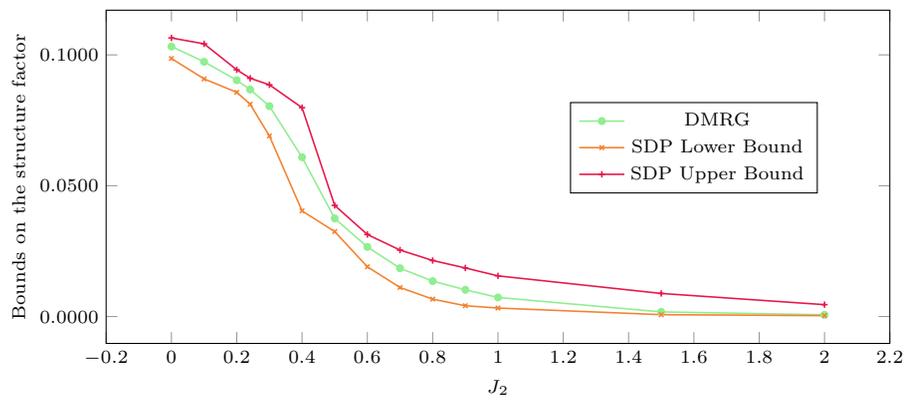
\begin{figure}	
\centering
\begin{tikzpicture}
\footnotesize
\scalefont{0.8} 
\begin{axis}[
sharp plot, 
xmode=normal,
width=15cm, height=8cm,  
yticklabel ={\pgfmathprintnumber[fixed, fixed zerofill, precision=4]{\tick}},
xlabel = $J_2$,
ylabel = Bounds on the structure factor,
xlabel near ticks,
ylabel near ticks,
legend style={at={(0.85,0.75)},anchor=south},
]

\addplot[semithick,mark=*,mark options={scale=0.6}, color=lightgreen] coordinates { 
     (0.0,     0.1031831)
     (0.1,     0.0973667)
     (0.2,     0.0903022)
     (0.24117, 0.0867427)
     (0.3,     0.0804070)
     (0.4,     0.0608781)
     (0.5,     0.0375000)
     (0.6,     0.0266019)
     (0.7,     0.0184667)
     (0.8,     0.0135322)
     (0.9,     0.0102386)
     (1.0,     0.0072953)
     (1.5,     0.0017755)
     (2.0,     0.0006739)
  };
\addlegendentry{DMRG}

\addplot[semithick,mark=x,mark options={scale=0.6}, color=color2] coordinates {
     (0.0,     0.0985998)
     (0.1,     0.0908153)
     (0.2,     0.0856754)
     (0.24117, 0.0811375)
     (0.3,     0.0690184)
     (0.4,     0.0403764)
     (0.5,     0.0324779)
     (0.6,     0.0190150)
     (0.7,     0.0111088)
     (0.8,     0.0066581)
     (0.9,     0.0041313)
     (1.0,     0.0032505)
     (1.5,     0.0007164)
     (2.0,     0.0003421)
  };
\addlegendentry{SDP Lower Bound}

\addplot[semithick,mark=+,mark options={scale=0.6}, color=color3] coordinates { 
     (0.0,     0.1064600)
     (0.1,     0.1042138)
     (0.2,     0.0942894)
     (0.24117, 0.0910491)
     (0.3,     0.0885589)
     (0.4,     0.0798086)
     (0.5,     0.0424656)
     (0.6,     0.0313481)
     (0.7,     0.0254093)
     (0.8,     0.0214150)
     (0.9,     0.0185697)
     (1.0,     0.0155411)
     (1.5,     0.0088262)
     (2.0,     0.0045430)
  };
\addlegendentry{SDP Upper Bound}

\end{axis}
\end{tikzpicture}
\caption{Bounds on the structure factor $S(\pi,\pi)$ for the $J_1$--$J_2$ Heisenberg chain ($N=40$).}
\label{fig:10}
\end{figure}  


\subsection{The 2D case}
Next, we turn to the more challenging 2D case. 
The SDP lower bounds on the ground-state energy per spin for the square-lattice Heisenberg model \eqref{model3} are displayed in Table~\ref{tab:B3-Energies}. For comparison, we also report the QMC values~\cite{sandvik1997} and the SDP lower bounds obtained in~\cite{wang2024certifying}; for $L=4,6$, the reference values are obtained by exact diagonalization. The relative SDP gap is normalized by the QMC values.
Table~\ref{tab:B3-Energies} shows that the new SDP relaxation substantially improves the ground-state energy bounds for every system size for which previous results are available. In particular, it is essentially exact for $L=4$, while for $L=6,8,10$ the relative gaps are reduced from $0.63\%$, $0.92\%$, $1.26\%$ to $0.30\%$, $0.43\%$, $0.48\%$, respectively. Moreover, the new method scales to lattices as large as $16\times16$, with the relative gap remaining below $0.7\%$ throughout. See also Figure~\ref{fig:7}.

\begin{table}[htbp]\label{tab:B3-Energies}
\caption{Bounds on the ground-state energy per spin of the square-lattice Heisenberg model with $L\times L$ spins.}
\centering
\begin{tabular}{c|c|cc|cc}
\Xhline{1pt}
\multirow{2}{*}{$L$}&\multirow{2}{*}{QMC}&\multicolumn{2}{c|}{SDP Old}&\multicolumn{2}{c}{SDP New}\\
&&value&gap&value&gap\\
\Xhline{1pt}
    4 &-0.701780&-0.703051&0.18\%&-0.701783&0.00\%\\
    6 &-0.678872&-0.683172&0.63\%&-0.680886&0.30\%\\
    8 &-0.673487&-0.679671&0.92\%&-0.676370&0.43\%\\
    10 &-0.671549&-0.680031&1.26\%&-0.674768&0.48\%\\
    12 &-0.670685&&&-0.674200&0.52\%\\
    14 &-0.670222&&&-0.674548&0.65\%\\
    16 &-0.669976&&&-0.674580&0.69\%\\
\Xhline{1pt}
\end{tabular}
\end{table}

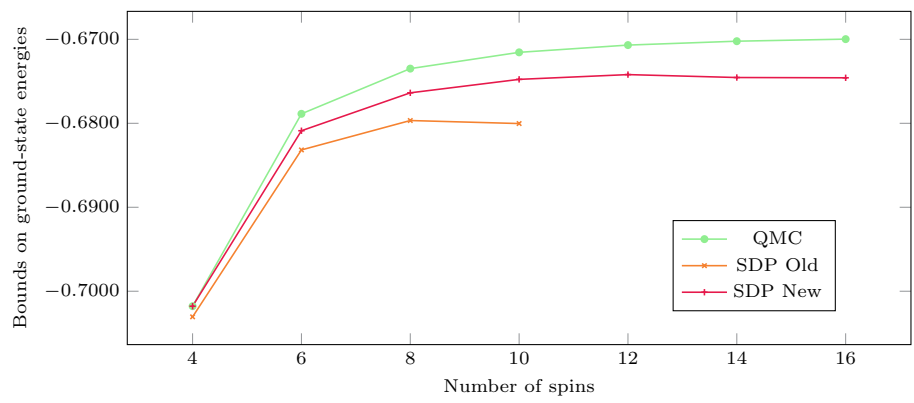
\begin{figure}	
\centering
\begin{tikzpicture}
\footnotesize
\scalefont{0.8} 
\begin{axis}[
sharp plot, 
xmode=normal,
width=15cm, height=8cm,  
yticklabel ={\pgfmathprintnumber[fixed, fixed zerofill, precision=4]{\tick}},
xlabel = Number of spins,
ylabel = Bounds on ground-state energies,
xlabel near ticks,
ylabel near ticks,
legend style={at={(0.85,0.1)},anchor=south},
]

\addplot[semithick,mark=*,mark options={scale=0.6}, color=lightgreen] coordinates { 
     (4,     -0.701780)
     (6,     -0.678872)
     (8,     -0.673487)
     (10,     -0.671549)
     (12,     -0.670685)
     (14,    -0.670222)
     (16,    -0.669976)
  };
\addlegendentry{QMC}

\addplot[semithick,mark=x,mark options={scale=0.6}, color=color2] coordinates { 
      (4,     -0.703050)
     (6,     -0.683171)
     (8,     -0.679670)
     (10,     -0.680030)
  };
\addlegendentry{SDP Old}

\addplot[semithick,mark=+,mark options={scale=0.6}, color=color3] coordinates { 
      (4,     -0.701783)
     (6,     -0.680886)
     (8,     -0.676370)
     (10,     -0.674768)
     (12,     -0.674200)
     (14,    -0.674548)
     (16,    -0.674580)
  };
\addlegendentry{SDP New}

\end{axis}
\end{tikzpicture}
\caption{Bounds on the ground-state energy of the square-lattice Heisenberg model with $L\times L$ spins.}
\label{fig:7}
\end{figure}  

Besides, Table~\ref{tab:B3-CL2} displays the lower and upper SDP bounds on the maximum-distance spin correlation $C(L/2,L/2)=\frac{1}{4}\langle \sigma_{(1,1)}^x \sigma_{(\frac{L}{2}+1,\frac{L}{2}+1)}^x \rangle$ for the square-lattice Heisenberg model \eqref{model3}. For comparison, we also report the QMC values~\cite{sandvik1997} and the widths of certified SDP intervals; for $L=4,6$, the reference values are obtained by exact diagonalization. 
Table~\ref{tab:B3-CL2} shows that the new SDP intervals consistently enclose the reference values and substantially improve upon the previous bounds. For $L=4$, the correlation is determined with essentially exact precision, while for $L=6$ and $L=8$ the interval widths are reduced by approximately $51\%$ and $55\%$, respectively. The intervals nevertheless widen as $L$ increases, and their lower bounds become negative for $L\geq12$, reflecting the difficulty of certifying this highly nonlocal correlation using a predominantly local monomial basis. 
One could expect that a different choice of the monomial basis may lead to better results for $C(L/2,L/2)$, a possibility that deserves further investigation. See also Figure~\ref{fig:8}.

\begin{table}[htb!]
\caption{Bounds on the maximum-distance spin correlation for the square-lattice Heisenberg model.}
    \label{tab:B3-CL2}
    \centering
    \begin{tabular}{c|c|c|c||c|c}
\Xhline{1pt}
$L$ & SDP Lower Bound &QMC& SDP Upper Bound&Width Old&Width New\\
\Xhline{1pt}
  4 & 0.059869 & 0.059872 & 0.059917 &0.0129& 0.0000 \\
6 & 0.039376 & 0.050856 & 0.057434 &0.0369& 0.0181 \\
8 & 0.025126 & 0.045867 & 0.055599 &0.0681& 0.0305 \\
10 & 0.014249 & 0.042851 & 0.054560 && 0.0403 \\
12 & -0.005382 & 0.040873 & 0.060823 && 0.0662 \\
14 & -0.021380 & 0.03945 & 0.063610 && 0.0850 \\
16 & -0.035612 & 0.03839 & 0.065971 && 0.1016 \\
\Xhline{1pt}
    \end{tabular}
\end{table}

\begin{figure}
\centering
\begin{tikzpicture}
\footnotesize
\scalefont{0.8} 
\begin{axis}[
sharp plot, 
xmode=normal,
width=15cm, height=8cm,  
ymin=-0.1, ymax=0.1,
yticklabel ={\pgfmathprintnumber[fixed, fixed zerofill, precision=4]{\tick}},
xlabel = $L$,
ylabel = Bounds on the spin correlations at maximum distance,
xlabel near ticks,
ylabel near ticks,
legend style={at={(0.2,0.05)},anchor=south},
]

\addplot[semithick,mark=*,mark options={scale=0.6}, color=lightgreen] coordinates { 
     (4,     0.059872)
     (6,     0.050856)
     (8,     0.045867)
     (10,    0.042851)
     (12,    0.040873)
     (14,    0.03945)
     (16,    0.03839)
  };
\addlegendentry{QMC}

\addplot[semithick,mark=x,mark options={scale=0.6}, color=color2] coordinates { 
     (4,     0.059869)
     (6,     0.039376)
     (8,     0.025126)
     (10,    0.014249)
     (12,    -0.005382)
     (14,    -0.021380)
     (16,    -0.035612)
  };
\addlegendentry{SDP Lower Bound (New)}

\addplot[semithick,mark=x,mark options={scale=0.6}, color=color1] coordinates { 
     (4,     0.052277)
     (6,     0.026268)
     (8,     0.002000)
  };
\addlegendentry{SDP Lower Bound (Old)}

\addplot[semithick,mark=+,mark options={scale=0.6}, color=color3] coordinates { 
     (4,     0.059917)
     (6,     0.057434)
     (8,     0.055599)
     (10,    0.054560)
     (12,    0.060823)
     (14,    0.063610)
     (16,    0.065971)
  };
\addlegendentry{SDP Upper Bound (New)}

\addplot[semithick,mark=+,mark options={scale=0.6}, color=bordeaux] coordinates { 
     (4,     0.065193)
     (6,     0.063146)
     (8,     0.070060)
  };
\addlegendentry{SDP Upper Bound (Old)}

\end{axis}
\end{tikzpicture}
\caption{Bounds on the maximum-distance spin correlation for the square-lattice Heisenberg model.}
\label{fig:8}
\end{figure}

Finally, consider the $10\times10$ square-lattice $J_1$--$J_2$ Heisenberg model \eqref{model4}. Table~\ref{tab:B4-L10} displays the SDP lower bounds on the ground-state energy per spin, together with the DMRG upper bounds~\cite{gongetal2014}, the NNQS upper bounds~\cite{chooetal2019}, and the SDP lower bounds obtained in~\cite{wang2024certifying}. The relative SDP gap is normalized by the NNQS bound or, when available, the DMRG bound.
Table~\ref{tab:B4-L10} shows that the new SDP relaxation consistently provides tighter lower bounds for the ground-state energy of the $10\times10$ $J_1$--$J_2$ Heisenberg model. The relative gaps remain small over the entire range of $J_2$, demonstrating that the proposed method retains good accuracy across different magnetic phases. See also Figure~\ref{fig:9}.

\begin{table}[htb!]
\caption{Bounds on the ground-state energy per spin for the square-lattice $J_1$--$J_2$ Heisenberg model ($L=10$).}
    \label{tab:B4-L10}
    \centering
    \begin{tabular}{c|c|c|cc|cc}
\Xhline{1pt}
\multirow{2}{*}{$J_2$}&\multirow{2}{*}{DMRG}&\multirow{2}{*}{NNQS}&\multicolumn{2}{c|}{SDP Old}&\multicolumn{2}{c}{SDP New}\\
&&&value&gap&value&gap\\
\Xhline{1pt}
    0.2      & &$-0.59275$ &-0.603083&1.74\%& $-0.599399$ &1.12\%\\
    0.4      & $-0.5253$ & $-0.52371$ &-0.537471&2.32\%& $-0.534018$& 1.66\%\\
    0.45      & $-0.5110$ & $-0.50905$ &-0.524642&2.66\%& $-0.521091$ &1.98\%\\
    0.5      & $-0.4988$ & $-0.49516$ &-0.514628&3.17\%& $-0.510813$ &2.41\%\\
    0.55    & $-0.4880$ & $-0.48277$ &-0.509309&4.36\%& $-0.504854$ &3.46\%\\
    0.6     && $-0.47604$ &-0.511361&7.42\%&  $-0.505984$ &6.29\%\\
    0.8     & & $-0.57383$ &-0.599453&4.46\%& $-0.593420$ 
    &3.41\%\\
    1.0     & & $-0.69636$ &-0.724752&4.08\%& $-0.717637$  &3.06\%\\
\Xhline{1pt}
\end{tabular}
\end{table}

\begin{figure}
\centering
\begin{tikzpicture}
\footnotesize
\scalefont{0.8} 
\begin{axis}[
sharp plot, 
xmode=normal,
width=15cm, height=8cm,  
yticklabel ={\pgfmathprintnumber[fixed, fixed zerofill, precision=4]{\tick}},
xlabel = $J_2$,
ylabel = Bounds on ground-state energies,
xlabel near ticks,
ylabel near ticks,
legend style={at={(0.15,0.05)},anchor=south},
]

\addplot[semithick,mark=*,mark options={scale=0.6}, color=lightgreen] coordinates { 
     (0.4,     -0.5253)
     (0.45,    -0.5110)
     (0.5,     -0.4988)
     (0.55,    -0.4880)
  };
\addlegendentry{DMRG}

\addplot[semithick,mark=*,mark options={scale=0.6}, color=darkblue] coordinates { 
     (0.2,     -0.59275)
     (0.4,     -0.52371)
     (0.45,    -0.50905)
     (0.5,     -0.49516)
     (0.55,    -0.48277)
     (0.6,     -0.47604)
     (0.8,     -0.57383)
     (1.0,     -0.69636)
  };
\addlegendentry{NNQS}

\addplot[semithick,mark=x,mark options={scale=0.6}, color=color2] coordinates { 
      (0.2,     -0.603083)
     (0.4,     -0.537471)
     (0.45,     -0.524642)
     (0.5,     -0.514628)
     (0.55,     -0.509309)
     (0.6,    -0.511361)
     (0.8,    -0.599453)
     (1.0,    -0.724752)
  };
\addlegendentry{SDP Old}

\addplot[semithick,mark=+,mark options={scale=0.6}, color=color3] coordinates { 
      (0.2,     -0.599399)
     (0.4,     -0.534018)
     (0.45,     -0.521091)
     (0.5,     -0.510813)
     (0.55,     -0.504854)
     (0.6,    -0.505984)
     (0.8,    -0.593420)
     (1.0,    -0.717637)
  };
\addlegendentry{SDP New}

\end{axis}
\end{tikzpicture}
\caption{Bounds on the ground-state energy for the square-lattice $J_1$--$J_2$ Heisenberg model ($L=10$).}
\label{fig:9}
\end{figure}
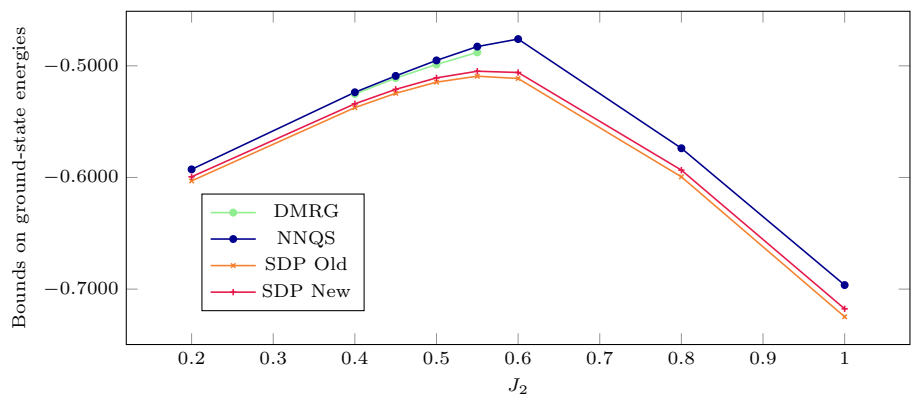  

\begin{remark}
For the $J_1$--$J_2$ Heisenberg chain with $0.1\le J_2\le0.9$, we empirically observed that imposing the PSD state optimality condition leads to some numerical issue when solving the corresponding SDP.\footnote{At an approximate ground state, the matrix in the PSD state optimality condition typically has many null or nearly null directions generated by stationarity and symmetry, which might yield an ill-conditioning SDP.} We thus remove the PSD state optimality condition from the SDP in this case.
Moreover, for the square-lattice $J_1$--$J_2$ Heisenberg model, we empirically observed that imposing either the PSD or the linear state optimality conditions leads to some numerical issue when solving the corresponding SDP. We thus omit both types of state optimality conditions.
\end{remark}

\section{Conclusions and future work}\label{sec:conclusions}
In this work, we have advanced the state of the art of the SDP relaxation method for certifying ground-state properties of quantum spin systems to unprecedented system sizes and accuracies, thus demonstrating its potential as a powerful tool for studying many-body physics. In particular, by extensively taking advantage of various structures of the system, we were able to get energy bounds for the square-lattice Heisenberg model of size up to $16 \times 16$. There are several directions for improving the method in future work:

$\bullet$ The DMRG is the prime tool for the study of 1D quantum many-body systems, featuring the coarse-graining maps. As shown in \cite{kull2024lower}, it is possible to combine the relaxation method with the DMRG to benefit from the coarse-graining maps.

$\bullet$ For Heisenberg models, the SU(2)-symmetry offers the moment matrix a block-diagonal structure in the corresponding symmetry-adapted basis. Incorporating this property into our method probably leads to a further reduction on the SDP size.

$\bullet$ Choosing which monomials to form the monomial basis can strongly affect the tightness of SDP relaxations. Ref.~\cite{MLbounds} proposed a deep-reinforcement-learning approach for this task. A more systematical study of this issue will be pursued in the future.

Moreover, it is also interesting to test the method developed here on systems that are notoriously difficult for variational calculations (see Ref.~\cite{wu2023variational}).

\section*{Acknowledgments}
This work was jointly funded by National Key R\&D Program of China under grant No.~2023YFA1009401, Natural Science Foundation of China under grant No.~12571333, Fundació Cellex, Fundació Mir-Puig, Generalitat de Catalunya (CERCA program), the Government of Spain (Severo Ochoa CEX2019-000910-S, QEC4QEA PCI2025-163167 and FUNQIP), the European Union (PASQuanS2.1 101113690, QEC4QEA 101194322, Quantera Veriqtas and COMPUTE 101017733), the ERC AdG CERQUTE, the AXA Chair in Quantum Information Science, and the ANR grant T-ERC QNET (ANR-24-ERCS-0008).

\bibliographystyle{siamplain}
\bibliography{references}

\newpage

\end{document}